%% file: ms.tex
\documentclass[acmsmall, nonacm]{acmart}

\newcommand{\iftwocol}[2]{#2}

\input{preamble}

\usemathitv

\makeatletter
\@ifclassloaded{acmart}{\newcommand*{\ifacmart}[2]{#1}}{\newcommand*{\ifacmart}[2]{#2}}
\makeatother

\togglefalse{short}
\ifshort{\setlength{\belowcaptionskip}{-6pt}}{}

\ifshort{\setlist{leftmargin=*}}{}

\begin{document}

\title{Optimal Multiserver Scheduling with Unknown Job Sizes \\ in Heavy Traffic}

\ifacmart{
  \acmclean
  \author{Ziv Scully}
  \affiliation{%
    \institution{Carnegie Mellon University}
    \department{Computer Science Department}
    \streetaddress{5000 Forbes Ave}
    \city{Pittsburgh}
    \state{PA}
    \postcode{15213}
    \country{USA}}
  \email{zscully@cs.cmu.edu}
  \author{Isaac Grosof}
  \affiliation{%
    \institution{Carnegie Mellon University}
    \department{Computer Science Department}
    \streetaddress{5000 Forbes Ave}
    \city{Pittsburgh}
    \state{PA}
    \postcode{15213}
    \country{USA}}
  \email{igrosof@cs.cmu.edu}
  \author{Mor Harchol-Balter}
  \affiliation{%
    \institution{Carnegie Mellon University}
    \department{Computer Science Department}
    \streetaddress{5000 Forbes Ave}
    \city{Pittsburgh}
    \state{PA}
    \postcode{15213}
    \country{USA}}
  \email{harchol@cs.cmu.edu}
}{
  \author{Ziv Scully}
  \ead{zscully@cs.cmu.edu}
  \author{Isaac Grosof}
  \ead{igrosof@cs.cmu.edu}
  \author{Mor Harchol-Balter}
  \ead{harchol@cs.cmu.edu}
  \address{%
    Carnegie Mellon University,
    Computer Science Department,
    5000 Forbes Ave,
    Pittsburgh,
    PA
    15213,
    USA}
}

\begin{abstract}
  We consider scheduling to minimize mean response time
  of the \mg{k} queue with unknown job sizes.
  In the single-server $k = 1$ case,
  the optimal policy is the \emph{Gittins} policy,
  but it is not known whether Gittins or any other policy
  is optimal in the multiserver case.
  Exactly analyzing the \mg{k} under any scheduling policy is intractable,
  and Gittins is a particularly complicated policy
  that is hard to analyze even in the single-server case.

  In this work we introduce \emph{monotonic Gittins} (\mgittins{}),
  a new variation of the Gittins policy,
  and show that it minimizes mean response time in the heavy-traffic \mg{k}
  for a wide class of finite-variance job size distributions.
  We also show that the
  \emph{monotonic shortest expected remaining processing time} (\mserpt{}) policy,
  which is simpler than \mgittins{},
  is a $2$\=/approximation for mean response time in the heavy traffic \mg{k}
  under similar conditions.
  These results constitute the most general optimality results to date
  for the \mg{k} with unknown job sizes.\ifshort{}{
  Our techniques build upon work by \citet{srpt_multiserver_grosof},
  who study simple policies, such as SRPT, in the \mg{k};
  \citet{fb_heavy_zwart, rmlf_zwart}, and \citet{srpt_heavy_zwart},
  who analyze mean response time scaling of simple policies
  in the heavy-traffic \mg{1};
  and \citet{m/g/1_gittins_aalto, mlps_gittins_aalto}
  and \citet{soap_scully, m-serpt_scully},
  who characterize and analyze the Gittins policy in the \mg{1}.}
\end{abstract}

\maketitle

\section{Introduction}
\label{sec:intro}

Scheduling to minimize mean response time\footnote{%
  A job's \emph{response time},
  also called \emph{sojourn time} or \emph{latency},
  is the amount of time between its arrival and its completion.}
of the \mg{k} queue
is an important problem in queueing theory.
The single-server $k = 1$ case has been well studied.
If the scheduler has access to each job's exact size,
the \emph{shortest remaining processing time} (SRPT) policy
is easily shown to be optimal
\citep{srpt_optimal_schrage}.
If the scheduler does not know job sizes,
which is very often the case in practical systems,
then a more complex policy called the \emph{Gittins} policy
is known to be optimal
\citep{m/g/1_gittins_aalto, mlps_gittins_aalto, book_gittins}.
The Gittins policy tailors its priority scheme to the job size distribution,
and it takes a simple form in certain special cases.
For example, for distributions with \emph{decreasing hazard rate} (DHR),
Gittins becomes the \emph{foreground-background} (FB) policy,\footnote{%
  FB is the policy that prioritizes the job of least age,
  meaning the job that has been served the least so far.
  It is also known as \emph{least attained service} (LAS).}
so FB is optimal in the \mg{1} for DHR job size distributions
\citep{fb_optimality_misra, m/g/1_gittins_aalto, mlps_gittins_aalto}.

In contrast to the \mg{1}, the \mg{k} with $k \geq 2$
has resisted exact analysis,
even for very simple scheduling policies.
As such, much less is known about minimizing mean response time in the \mg{k},
with the only nontrivial results holding under heavy traffic.\footnote{%
  Here ``heavy traffic'' refers to the limit as
  the system load approaches capacity for a fixed number of servers.}
For known job sizes,
recent work by \citet{srpt_multiserver_grosof} shows that
a multiserver analogue of SRPT is optimal in the heavy-traffic \mg{k}.
For unknown job sizes, \citet{srpt_multiserver_grosof} address only the case of
DHR job size distributions,
showing that a multiserver analogue of FB
is optimal in the heavy-traffic \mg{k}.\footnote{%
  Both the SRPT and FB optimality results of \citet{srpt_multiserver_grosof}
  hold under technical conditions similar to finite variance.}
But in general, optimal scheduling is an open problem for unknown job sizes,
even in heavy traffic.
We therefore ask:
\begin{quote}
  \emph{What scheduling policy minimizes mean response time
    in the heavy-traffic \mg{k}
    with unknown job sizes and general job size distribution?}
\end{quote}

This is a very difficult question.
In order to answer it,
we draw upon several recent lines of work in scheduling theory.
\begin{itemize}
\item
  As part of their heavy-traffic optimality proofs,
  \citet{srpt_multiserver_grosof} use a tagged job method
  to stochastically bound
  \mg{k} response time under each of SRPT and FB
  relative to \mg{1} response time (\cref{fig:mgk}) under the same policy.
\item
  \citet{srpt_heavy_zwart} and \citet{fb_heavy_zwart}
  characterize the heavy-traffic scaling of
  \mg{1} mean response time under SRPT and FB, respectively.
\item
  \Citet{m-serpt_scully} show that a policy called
  \emph{monotonic shortest expected remaining processing time} (\mserpt{}),
  which is considerably simpler than Gittins,
  has \mg{1} mean response time within a constant factor of that of Gittins.
\end{itemize}
While these prior results do not answer the question on their own,
together they suggest a plan of attack for proving optimality
in the heavy-traffic \mg{k}.

When searching for a policy to minimize mean response time,
a natural candidate is a multiserver analogue of Gittins.
As a first step,
one might hope to use the tagged job method of \citet{srpt_multiserver_grosof}
to stochastically bound \mg{k} response time under Gittins
relative to \mg{1} response time.
Unfortunately, the tagged job method
does not apply to multiserver Gittins,
because it relies on both stochastic and worst-case properties of the scheduling policy,
whereas Gittins has poor worst-case properties.

One of our key ideas is to introduce a new variant of Gittins,
called \emph{monotonic Gittins} (\mgittins{}),
that has better worst-case properties than Gittins
while maintaining similar stochastic properties.
This allows us to generalize the tagged job method \citep{srpt_multiserver_grosof}
to \mgittins{},
thus bounding its \mg{k} response time relative to its \mg{1} response time.

Our \mg{k} analysis of \mgittins{}
reduces the question of whether \mgittins{} is optimal in the heavy-traffic \mg{k}
to analyzing the heavy-traffic scaling of \mgittins{}'s \mg{1} mean response time.
However, there are no heavy-traffic scaling results for the \mg{1}
under policies other than
SRPT \citep{srpt_heavy_zwart},
FB \citep{fb_heavy_zwart},
\emph{first-come, first served} (FCFS) \citep{fcfs_heavy_dist_kollerstrom, fcfs_heavy_mean_kollerstrom},
and a small number of other simple policies
\citep{rmlf_zwart, two_classes_chen}.
To remedy this, we derive heavy-traffic scaling results for \mgittins{} in the \mg{1}.
It turns out that analyzing \mgittins{} directly is very difficult.
Fortunately, \mgittins{} has a simpler cousin, \mserpt{},
which \citet{m-serpt_scully} introduce and analyze.
We analyze \mserpt{} in heavy traffic
as a key stepping stone in our \ifshort{}{heavy-traffic }analysis of \mgittins{}.

This paper makes the following contributions:
\begin{itemize}
\item
  We introduce the \mgittins{} policy and prove that
  it minimizes mean response time in the heavy-traffic \mg{k}
  for a large class of finite-variance job size distributions
  (\cref{thm:mgittins_opt}).
\item
  We also prove that the simple and practical \mserpt{} policy
  is a $2$\=/approximation for mean response time in the heavy-traffic \mg{k}
  for a large class of finite-variance job size distributions
  (\cref{thm:mserpt_approx}).
\item
  We characterize the heavy-traffic scaling
  of mean response time in the \mg{1}
  under Gittins, \mgittins{}, and \mserpt{}
  (\cref{thm:heavy}).
\end{itemize}
\Cref{sec:main_results} formally states these results
and compares them to prior work.
Their proofs rely on a large collection of intermediate results,
which we outline in detail in \cref{sec:overview}
and prove in \cref{sec:mgk, sec:rank_bounds, sec:heavy}.

\section{Preliminaries}

We consider an \mg{k} queue with arrival rate~$\lambda$
and job size distribution~$X$.
Each of the $k$ servers has speed $1/k$,
so regardless of the number of servers,
the total service rate is~$1$
and the system load is $\rho = \lambda\E{X}$.
This allows us to easily compare the \mg{k} system
to a single-server \mg{1} system, as illustrated in \cref{fig:mgk}.
We assume a preempt-resume model with no preemption overhead.
This means that a single-server \mg{1} system
can simulate any \mg{k} policy
by time-sharing between $k$ jobs.

Throughout this paper we consider the $\rho \to 1$ or \emph{heavy-traffic} limit.
This is the $\lambda \to 1/\E{X}$ limit
with the job size distribution~$X$ and number of servers~$k$ held constant.

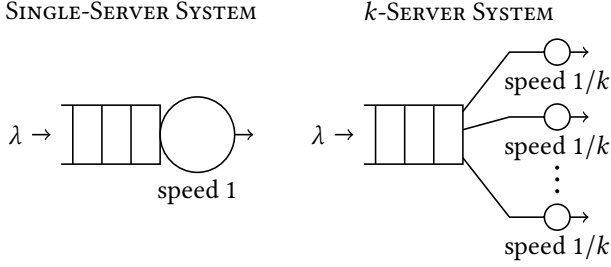
\begin{figure}
  \centering
  \input{fig_mgk}
  \captionsqueeze
  \caption{Single-Server and $k$-Server Systems}
  \label{fig:mgk}
\end{figure}

We write $F$ for the cumulative distribution function of~$X$
and $\F{x} = 1 - F(x)$ for its tail.
We assume that $X$ has a continuous, piecewise-monotonic\footnote{%
  A function is piecewise-monotonic if, roughly speaking,
  it switches between increasing and decreasing finitely many times
  in any compact interval.}
hazard rate\ifshort{ $h(x) = F'(x)/\F{x}$.}{
\begin{align*}
  h(x) = \frac{\frac{\d}{\d{x}} F(x)}{\F{x}}.
\end{align*}}
We also frequently work with the expected remaining size of a job at age~$a$,
which is $\E{X - a \given X > a}$.
We assume it, too, is continuous and piecewise-monotonic as a function of~$a$.
\ifshort{\footnote}{\par}{%
The above assumptions on hazard rate and expected remaining size
are not restrictive and serve primarily to simplify presentation.\ifshort{}{
It is very likely that our proofs can be generalized to relax them.}}

\subsection{SOAP Policies and Rank Functions}
\label{sub:rank_functions}

All of the scheduling policies considered in this work
are in the class of \emph{SOAP policies} \citep{soap_scully},
generalized to a multiserver setting.
In a single-server setting,
a SOAP policy~\generic{} is specified by a \emph{rank function}
\ifshort{\(}{\begin{align*}}
  \rank{\generic{}}{} : \R_+ \to \R
\ifshort{\)}{\end{align*}}
which maps a job's \emph{age},
namely the amount of service it has received so far,
to its \emph{rank}, or priority level.
Single-server SOAP policies work by always serving
the job of \emph{minimal rank},
breaking ties in FCFS fashion.\footnote{%
  The full SOAP class allows a job's rank to depend
  on both its age and its ``static'' characteristics,
  such as its size or class,
  but we do not use this generality in this paper.}

As an example, FB is a SOAP policy with $\rank{\fb{}}{a} = a$.
\ifshort{L}{Because l}ower age corresponds to lower rank,
\ifshort{so }{}FB prioritizes the job of least age.\footnote{%
  When multiple jobs are tied for least age,
  FB shares the server among all such jobs
  because the rank function is increasing.
  See \citet[Appendix~B]{soap_scully} for details.}

A multiserver SOAP policy uses the same rank function as its single-server analogue.
The only difference is that the system can serve up to $k$~jobs,
so a multiserver SOAP policy \ifshort{%
serves the $k$ jobs of minimal rank,
breaking ties in FCFS fashion.
}{%
works as follows:
\begin{itemize}
\item
  If there are at most $k$ jobs in the system,
  serve all of them.
\item
  If there are more than $k$ jobs in the system,
  serve the $k$ jobs of minimal rank,
  breaking ties in FCFS fashion.
\end{itemize}}
We often compare the $k$-server variant of a policy~\generic{}
to its single-server analogue.
When it is necessary to distinguish between them,
we write \generic{k} for the $k$-server version of a policy,
so \generic{1} is the single-server version.
We write $\response{\generic{k}}{x}$
for the size-conditional response time distribution of jobs of size~$x$
under \generic{k},
and we write $\response{\generic{k}}{}$
for the overall response time distribution.

\ifshort{We primarily consider four policies:}{There are four main policies we consider in this work:}
\ifshort{
    shortest expected remaining processing time (SERPT),
    monotonic SERPT (\mserpt{}),
    Gittins,
    and monotonic Gittins (\mgittins{}).
}{
    SERPT, \mserpt{}, Gittins, and \mgittins{}.}
None of the policies need job size information,
but each uses the job size distribution to tune its rank function.
\ifshort{}{As an example, \cref{fig:rank_examples} shows
the four rank functions for a bounded distribution with nonmonotonic hazard rate.

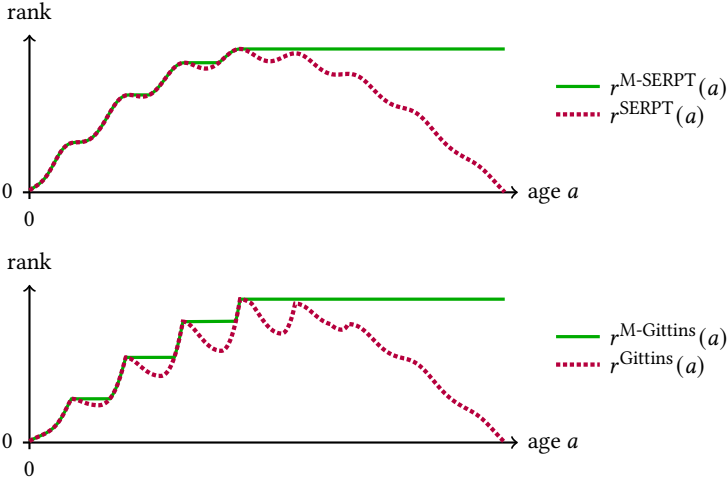
\begin{figure}
  \centering
  \input{fig_serpt_example}
  \input{fig_gittins_example}
  \captionsqueeze
  \caption{Rank Function Examples}
  \label{fig:rank_examples}
\end{figure}}
\ifshort{The rank functions of the policies are as follows:
\begin{definition}
  \label{def:serpt}
  \label{def:mserpt}
  \label{def:gittins}
  \label{def:mgittins}
    \begin{align*}
    \rank{\serpt{}}{a}
        &= \E{X - a \given X > a} \\
    \rank{\mserpt{}}{a}
        &= \max_{b \in [0, a]} \rank{\serpt{}}{b} \\
    \rank{\gittins{}}{a}
        &= \inf_{b > a} \frac{\E{\min\{X, b\} - a \given X > a}}{\P{X \leq b \given X > a}} \\
    \rank{\mgittins{}}{a}
        &= \max_{b \in [0, a]} \rank{\gittins{}}{b}.
    \end{align*}
\end{definition}
}{
\begin{definition}
  \label{def:serpt}
  The \emph{shortest expected remaining processing time} (SERPT) policy is
  the SOAP policy with rank function
  \begin{align*}
    \rank{\serpt{}}{a}
    = \E{X - a \given X > a}
    \ifshort{}{= \frac{\int_a^\infty \F{t} \d{t}}{\F{a}}}.
  \end{align*}
  As a reminder, lower rank means better priority,
  so, as hinted by its name, SERPT prioritizes the job
  of least expected remaining size.
\end{definition}

\begin{definition}
  \label{def:mserpt}
  The \emph{monotonic SERPT} (\mserpt{}) policy is
  the SOAP policy with monotonic rank function
  \begin{align*}
    \rank{\mserpt{}}{a}
    = \max_{b \in [0, a]} \rank{\serpt{}}{b}.
  \end{align*}
\end{definition}

\begin{definition}
  \label{def:gittins}
  The \emph{Gittins} policy is
  the SOAP policy with rank function
  \begin{align*}
    \rank{\gittins{}}{a}
    = \inf_{b > a} \frac{\E{\min\{X, b\} - a \given X > a}}{\P{X \leq b \given X > a}}
    \ifshort{}{= \inf_{b > a} \frac{\int_a^b \F{t} \d{t}}{\F{a} - \F{b}}}.
  \end{align*}
\end{definition}

\begin{definition}
  \label{def:mgittins}
  The \emph{monotonic Gittins} (\mgittins{}) policy is
  the SOAP policy with monotonic rank function
  \begin{align*}
    \rank{\mgittins{}}{a}
    = \max_{b \in [0, a]} \rank{\gittins{}}{b}.
  \end{align*}
\end{definition}
}
\ifshort{}{The \mgittins{} and \mserpt{} policies,
which both have monotonic rank functions,
are the primary focus of this paper.
Some of our intermediate results apply more broadly
to any policy with a monotonic rank function.}

\begin{definition}
  A SOAP policy~\generic{} is \emph{monotonic}
  if its rank function is nondecreasing,
  meaning $\rank{\generic{}}{a} \leq \rank{\generic{}}{b}$
  for all ages $a < b$.\footnote{%
    The nonincreasing case\ifshort{
      includes only FCFS and is thus less interesting.%
    }{
      is less interesting,
      because all nonincreasing rank functions encode FCFS.}}
\end{definition}

\ifshort{%
  SERPT and Gittins can have nonmonotonic rank functions.
}{%
  \Cref{fig:rank_examples} shows
  the SERPT, \mserpt{}, Gittins, and \mgittins{} rank functions
  for a bounded distribution with nonmonotonic hazard rate.
  Notice that SERPT and Gittins are not monotonic.%
}
This makes it hard to analyze their \mg{k} response time\ifshort{}{
(\cref{app:mgk_nonmonotonic})}.
In contrast, the \mserpt{} and \mgittins{} are monotonic:
their rank functions alternate between constant regions
and strictly increasing regions.

\ifshort{%
  While the rank functions of Gittins and SERPT may not be monotonic,
  they are still well behaved, namely continuous and piecewise monotonic,
  under our assumptions on the job size distribution
  \citep[Theorem~1]{mlps_gittins_aalto}.%
}{%
While the rank functions of Gittins and SERPT may not be monotonic,
they are still well behaved under our assumptions on the job size distribution.

\begin{lemma}
  \label{lem:rank_continuous}
  Under the assumption that the job size distribution~$X$
  has continuous and piecewise-monotonic
  hazard rate and expected remaining size functions,
  each of $\rank{\serpt{}}{}$, $\rank{\mserpt{}}{}$,
  $\rank{\gittins{}}{}$, and $\rank{\mgittins{}}{}$
  is continuous and piecewise-monotonic.
\end{lemma}

\begin{proof}
  It suffices to prove the claims for $\rank{\serpt{}}{}$
  and $\rank{\gittins{}}{}$.
  The claim for $\rank{\serpt{}}{}$ is exactly our assumption on
  expected remaining size,
  and the claim for $\rank{\gittins{}}{}$ is a known result
  \citep[Theorem~1]{mlps_gittins_aalto}.
\end{proof}}

\subsection{Job Size Distribution Classes}
\label{sub:distributions}

We consider several classes of job size distributions in this paper.
We briefly describe each class before giving the formal definitions.
\begin{itemize}
\item
  The $\OR{1}{\infty}$ class (\cref{def:or})
  contains, roughly speaking, distributions with Pareto-like tails.\ifshort{}{
  \begin{itemize}
  \item}
    We focus especially on the $\OR{2}{\infty}$ subclass,
    all members of which have finite variance.\ifshort{}{
  \end{itemize}}
\item
  The $\Gumbel$ class\ifshort{}{ (\cref{def:gumbel})}
  contains, roughly speaking, distributions with smooth tails
  that are lighter than Pareto tails.
  It includes\ifshort{}{, among others,}
  exponential, normal, log-normal, Weibull, and Gamma distributions.
\item
  The $\QDHR$ and $\QIMRL$ classes (\ifshort{\cref{def:qdhr_qimrl}}{\cref{def:qdhr, def:qimrl}})
  are relaxations of the well-known \emph{decreasing hazard rate} ($\DHR$)
  and \emph{increasing mean residual lifetime} ($\IMRL$) classes
  \citep{m/g/1_gittins_aalto, mlps_gittins_aalto, fb_optimality_misra,
    \ifshort{}{book_shaked,}fb_nonoptimality_aalto, dhr_dmrl_optimality_righter,
    dhr_ihr_optimality_righter, mlps_delay_aalto}.
  $\QDHR$ contains distributions whose hazard rate
  is roughly decreasing with age,
  even if it is not perfectly monotonic,
  and $\QIMRL$ contains distributions with
  roughly increasing expected remaining size.\ifshort{}{%
  \begin{itemize}
  \item
    We focus especially on the subclasses
    $\Gumbel \cap \QDHR$ and $\Gumbel \cap \QIMRL$.
  \end{itemize}}
\item
  The $\ENBUE$ class (\cref{def:enbue})
  is a relaxation of the well-known \emph{new better than used in expectation} ($\NBUE$) class
  \citep{m/g/1_gittins_aalto, mlps_gittins_aalto\ifshort{}{,book_shaked}}.\ifshort{}{\footnote{%
    Because the $\NBUE$ terminology originates in reliability analysis,
    the word ``better'' here means ``longer''.}}
  It contains distributions whose expected remaining size
  reaches a global maximum at some age\ifshort{,
  which includes all distributions with bounded support.}{.
  \begin{itemize}
  \item
    We focus especially on the $\Bounded$ subclass,
    which contains all bounded distributions.
  \end{itemize}}
\end{itemize}

These classes play two different roles in our analysis.
\begin{itemize}
\item
  Some of the classes broadly characterize
  the asymptotic behavior of the tail~$\F{}$.
  These include $\OR{1}{\infty}$, $\Gumbel$, and $\ENBUE$.
  Virtually all job size distributions of interest are in one of these classes,
  so requiring membership in one of them, as in \cref{thm:heavy},
  should not be viewed as a major restriction.
\item
  Some of the classes impose additional conditions on the job size distribution
  that help us bound the \mgittins{} and \mserpt{} rank functions (\cref{sec:rank_bounds}).
  These include $\QDHR$, $\QDHR$, and $\Bounded$.
  While these classes are much broader than those previously studied (\cref{sub:prior_work}),
  they do not cover all distributions of interest.
  Requiring membership in one of them, as in \cref{thm:mgittins_opt, thm:mserpt_approx},
  represents a genuine restriction.
\end{itemize}

\begin{definition}
  \label{def:or}
  A function~$f$ is \emph{$O$-regularly varying}
  if there exist exponents $\beta \geq \alpha > 0$
  along with constants $C_0, x_0 > 0$
  such that for all $y \geq x \geq x_0$,
  \begin{align*}
    \frac{1}{C_0} \gp*{\frac{y}{x}}^{-\beta}
    \leq \frac{f(y)}{f(x)}
    \leq C_0 \gp*{\frac{y}{x}}^{-\alpha}.
  \end{align*}
  We write $\OR{\alpha_0}{\beta_0}$ for the set of
  $O$-regularly varying functions
  where the exponents $\alpha$ and~$\beta$ above may be chosen such that
  $\alpha_0 < \alpha \leq \beta < \beta_0$.\footnote{%
    This is not the standard definition of $O$-regular variation,
    but it is equivalent to it \citep[Section~2.2.1]{book_bingham}.
    Specifically, our $\OR{\alpha_0}{\beta_0}$
    contains the $O$-regularly varying functions
    whose Matuszewska indices are in the interval $(-\beta_0, -\alpha_0)$.}
  We use the same $\OR{\alpha_0}{\beta_0}$ notation to represent
  the class of distributions
  whose tails are in $\OR{\alpha_0}{\beta_0}$.
\end{definition}

\ifshort{%
\begin{definition}
  \label{def:qdhr_qimrl}
  A job size distribution is in the
  \emph{quasi-decreasing hazard rate} ($\QDHR$) class,
  if there exist
  a strictly increasing function $m : \R_+ \to \R_+$,
  an exponent $\gamma \geq 1$,
  and constants $C_0, x_0 > 0$ such that for all $x \geq x_0$,
  \begin{align*}
      m(x) \leq \ifshort{h(x)^{-1}}{\frac{1}{h(x)}} \leq m(C_0 x^\gamma).
  \end{align*}
  Similarly, a distribution is in the
  \emph{quasi-increasing mean residual lifetime} ($\QIMRL$) class
  if under the same conditions,
  \begin{align*}
    m(x) \leq \E{X - x \given X > x} \leq m(C_0 x^\gamma).
  \end{align*}
\end{definition}
}{
\begin{definition}
  \label{def:qdhr}
  A job size distribution is in the
  \emph{quasi-decreasing hazard rate} class, denoted $\QDHR$,
  if there exist
  a strictly increasing function $m : \R_+ \to \R_+$,
  an exponent $\gamma \geq 1$,
  and constants $C_0, x_0 > 0$ such that for all $x \geq x_0$,
  \begin{align*}
    m(x) \leq \frac{1}{h(x)} \leq m(C_0 x^\gamma).
  \end{align*}
\end{definition}

\begin{definition}
  \label{def:qimrl}
  A job size distribution is in the
  \emph{quasi-increasing mean residual lifetime} class, denoted $\QIMRL$,
  if there exist
  a strictly increasing function $m : \R_+ \to \R_+$,
  an exponent $\gamma \geq 1$,
  and constants $C_0, x_0 > 0$ such that for all $x \geq x_0$,
  \begin{align*}
    m(x) \leq \E{X - x \given X > x} \leq m(C_0 x^\gamma).
  \end{align*}
\end{definition}
}

\begin{definition}
  \label{def:enbue}
  A job size distribution is in the
  \emph{eventually new better than used in expectation} class,
  denoted $\ENBUE$,
  if there exists an age $a_* \geq 0$ at which a job's expected remaining size
  reaches a global maximum,
  meaning that for all $x \neq a_*$,
  \begin{align*}
    \E{X - a_* \given X > a_*} \geq \E{X - x \given X > x}.
  \end{align*}
  \ifshort{$\ENBUE$ contains $\Bounded$, distributions with bounded support.}{}
\end{definition}

\ifshort{}{\begin{definition}
  A job size distribution is in the \emph{bounded} class, denoted $\Bounded$,
  if there exists $x_{\max} < \infty$ such that $\F{x_{\max}} = 0$.
\end{definition}}

\ifshort{}{\begin{definition}
  \label{def:gumbel}
  A job size distribution is said to be in the
  \emph{Gumbel domain of attraction}, denoted $\Gumbel$,
  under certain conditions specified in extreme value theory
  \citep{book_resnick}.
\end{definition}}

The exact characterization of $\Gumbel$\ifshort{,
  which comes from extreme value theory \citep{book_resnick},}{}
is outside the scope of this paper.
The most important property is that distributions in $\Gumbel$
are lighter-tailed than all Pareto distributions.\ifshort{
Its proof follows from a known characterization of $\Gumbel$
\citep[Proposition~1.4]{book_resnick}.}{}

\begin{lemma}
  \label{lem:F_gumbel}
  If $X \in \Gumbel$, then $\F{x} = o(x^{-\alpha})$ for all $\alpha > 0$.
\end{lemma}

\ifshort{}{\begin{proof}
  The result follows from a known characterization of $\Gumbel$
  \citep[Proposition~1.4]{book_resnick}.
\end{proof}}

\section{Main Results}
\label{sec:main_results}

We now present our main results,\ifshort{
beginning with our heavy-traffic \mg{k} optimality result.}{
explaining how they relate to prior work in \cref{sub:prior_work}.
We begin with our heavy-traffic \mg{k} optimality result.}

\begin{theorem}
  \label{thm:mgittins_opt}
  \ifshort{We have}{In an \mg{k}, if
  \begin{align*}
    X \in \OR{2}{\infty}
      \cup (\Gumbel \cap \QDHR)
      \cup \Bounded,
  \end{align*}
  then}
  \ifshort{\(}{\begin{align*}}
    \lim_{\rho \to 1} \ifshort{\slashfrac}{\frac}{\E{\response{\mgittins{k}}{}}}{\E{\response{\gittins{1}}{}}} = 1\ifshort{}{.}
  \ifshort{\)}{\end{align*}}
  \ifshort{if $X$ is in $\OR{2}{\infty}$, $\Gumbel \cap \QDHR$, or $\Bounded$.}{}
  In such cases, \mgittins{k} is optimal for mean response time in heavy traffic.
\end{theorem}

The \mgittins{} policy is based on the Gittins policy,
which is somewhat complex to describe and compute.
Fortunately, the \mserpt{} policy,
which can be much simpler to compute \citep{m-serpt_scully},
also performs well in the heavy-traffic \mg{k}.

\begin{theorem}
  \label{thm:mserpt_approx}
  \ifshort{We have}{In an \mg{k}, if
  \begin{align*}
    X \in \OR{2}{\infty}
      \cup (\Gumbel \cap (\QDHR \cup \QIMRL))
      \cup \Bounded,
  \end{align*}
  then}
  \ifshort{\(}{\begin{align*}}
    \lim_{\rho \to 1} \ifshort{\slashfrac}{\frac}{\E{\response{\mserpt{k}}{}}}{\E{\response{\gittins{1}}{}}} \leq 2\ifshort{}{.}
  \ifshort{\)}{\end{align*}}
  \ifshort{if $X$ is in $\OR{2}{\infty}$, $\Gumbel \cap (\QDHR \cup \QIMRL)$, or $\Bounded$.}{}
  In such cases, \mserpt{k} is a $2$\=/approximation for mean response time in heavy traffic.
\end{theorem}

\Cref{thm:mgittins_opt, thm:mserpt_approx} apply to a broad class of
finite-variance job size distributions.
Roughly speaking, $\OR{2}{\infty}$ covers heavy-tailed distributions,
and $\Gumbel$ covers non-heavy-tailed distributions that are unbounded
(\cref{sub:distributions}).
Assuming membership in these sets is standard for heavy-traffic analysis
\citep{fb_heavy_zwart}.
The main restriction the results impose is on $\Gumbel$ distributions,
for which we additionally require membership in $\QDHR$ or $\QIMRL$.
While slightly relaxing this restriction is possible,\footnote{%
  For example, we only need the $\QDHR$ and $\QIMRL$ assumptions to prove
  \cref{thm:yz_bound_mserpt_qimrl, thm:yz_bound_mgittins_qdhr},
  so we could instead assume the results of those theorems.}
removing it entirely appears to be very difficult (\cref{sec:conclusion}).

A key step in the proofs of \cref{thm:mgittins_opt, thm:mserpt_approx}
is analyzing \mgittins{} and \mserpt{} in the heavy-traffic \mg{1}.
This analysis is itself a new result of independent interest.
Notably, it extends to ordinary Gittins in addition to \mgittins{}\ifshort{}{,
thus characterizing the optimal heavy-traffic scaling
attainable by any scheduling policy in the setting of unknown job sizes}.

\begin{theorem}
  \label{thm:heavy}
  Let \generic{1} be one of \gittins{1}, \mgittins{1}, or \mserpt{1}.
  If $X \in \OR{1}{2}$, then in the $\rho \to 1$ limit,
  \begin{align*}
    \E{\response{\generic{1}}{}}
    = \Theta\gp*{\ifshort{-\log(1 - \rho)}{\log\frac{1}{1 - \rho}}}
  \end{align*}
  and if $X \in \OR{2}{\infty} \cup \Gumbel \cup \ENBUE$, then
  \begin{align*}
    \E{\response{\generic{1}}{}}
    = \Theta\ifshort{\gp[\Big]}{\gp*}{\ifshort{\gp[\Big]}{\frac{1}}{(1 - \rho) \cdot \rank{\mserpt{}}[\big]{\Ginv{1 - \rho}}}\ifshort{\sp{-1}}{}},
  \end{align*}
  where $\Ginv{}\esub$ is the inverse of the tail of the excess of~$X$, namely
  \begin{align*}
    \G{x} = \frac{1}{\E{X}} \int_x^\infty \F{t} \d{t}.
  \end{align*}
\end{theorem}

\subsection{Relationship to Prior Work}
\label{sub:prior_work}

\Cref{thm:mgittins_opt}
is the first result proving optimality of a scheduling policy
in the heavy-traffic \mg{k}
with unknown job sizes and general job size distribution.
As mentioned in \cref{sec:intro},
the only prior results of this type were shown by \citet{srpt_multiserver_grosof},
who prove similar results for SRPT and FB,
that latter for \emph{decreasing hazard rate} ($\DHR$) job size distributions.
\ifshort{\par}{\begin{itemize}
\item}
  SRPT was shown to be optimal in the heavy-traffic \mg{k}
  for job size distributions whose tail has
  upper Matuszewska index less than~$-2$
  \citep[Theorem~6.1]{srpt_multiserver_grosof},
  which corresponds to satisfying the upper bound
  in \cref{def:or} for some $\alpha > 2$.
  This is somewhat broader than the precondition of \cref{thm:mgittins_opt},
  though it is still limited to finite-variance distributions.
  \ifshort{}{\begin{itemize}
  \item}
    Given that SRPT is designed for known job sizes
    while \mgittins{} is designed for unknown job sizes,
    \cref{thm:mgittins_opt} complements the prior SRPT results.
  \ifshort{\par}{\end{itemize}
\item}
  FB was shown to be optimal in the heavy-traffic \mg{k}
  for job size distributions in the class
  $\DHR \cap (\OR{2}{\infty} \cup \Gumbel)$
  \citep[Theorem~7.13]{srpt_multiserver_grosof}.\footnote{%
    While \citet[Theorem~7.13]{srpt_multiserver_grosof} claim that
    this result applies to all distributions in $\DHR$
    with upper Matuszewska index less than~$-2$,
    their proof incorrectly cites the preconditions of
    results of \citet{fb_heavy_zwart}.\ifshort{}{
    Correcting the precondition narrows the result to what we state here.}}
  The $\DHR$ class is much more restrictive than $\QDHR$,
  so this is much narrower than the precondition of \cref{thm:mgittins_opt}.
  \ifshort{}{\begin{itemize}
  \item}
    Given that FB is equivalent to \mgittins{} in the $\DHR$ case
    \citep{m/g/1_gittins_aalto, mlps_gittins_aalto},
    \cref{thm:mgittins_opt} subsumes the prior FB results.
  \ifshort{}{\end{itemize}
\end{itemize}}

\ifshort{}{There is another result that follows from two prior works
that complements \cref{thm:mgittins_opt},
although to the best of our knowledge it has never been explicitly stated.
\Citet{fcfs_heavy_dist_kollerstrom, fcfs_heavy_mean_kollerstrom}
shows that under FCFS,
the mean response times in the \mg{1} and \mg{k} converge.
This means that if \gittins{} and \mgittins{}
happen to be equivalent to FCFS for a given job size distribution,
then FCFS minimizes mean response time in the heavy-traffic \mg{k}.
\Citet{m/g/1_gittins_aalto, mlps_gittins_aalto}
show this occurs exactly for job size distributions in the
\emph{new better than used in expectation} ($\NBUE$) class,
which includes some distributions that \cref{thm:mgittins_opt} does not cover.

Finally, versions of the Gittins policy have been shown to be heavy-traffic optimal
for two discrete-state versions of the \mg{k} queue
\citep{m/m/k_gittins_glazebrook, m/g/k_gittins_glazebrook}.
These models support some features our model does not,
such as multiple job classes,
but discretizing the state space imposes some limitations.
Specifically, \citet{m/m/k_gittins_glazebrook} require each job to be composed of phases
where each phase has exponentially distributed size;
and \citet{m/g/k_gittins_glazebrook} allows nonexponential job size distributions
but discretizes time
and additionally requires $\ENBUE$ job size distributions (\cref{def:enbue}).
In contrast, \cref{thm:mgittins_opt} applies to heavy-tailed
and other non-$\ENBUE$ job size distributions
that are of practical importance in computer systems
\citep{pareto_tails_crovella, pareto_tails_williamson, real_tails_survey_park,
  pareto_tails_harchol-balter}.}

\Cref{thm:mserpt_approx} shows that a simple scheduling policy,
namely \mserpt{},
has mean response time within a constant factor of optimal
in the heavy-traffic \mg{k}
with unknown job sizes and general job size distribution.
Specifically, we show \mserpt{} is a $2$\=/approximation.
This complements the result of \citet{m-serpt_scully},
who show that in the \mg{1},
\mserpt{} is a $5$\=/approximation
for \mg{1} mean response time at all loads.
Our result is tighter and applies to multiserver systems, not just single-server systems,
but it applies only in heavy traffic.
The techniques we introduce could be useful for tightening the upper bound
on \mserpt{}'s \mg{1} approximation ratio,
which is conjectured to be~$2$ \citep{m-serpt_scully}.

\Cref{thm:heavy} characterizes the heavy-traffic scaling
of \mg{1} mean response time
under Gittins, \mgittins{}, and \mserpt{}.
There are three other policies
whose heavy-traffic scaling has been characterized:
FB, SRPT, and a policy called
\emph{randomized multilevel feedback} (RMLF) \citep{rmlf_pruhs, rmlf_leonardi}.
We now compare \cref{thm:heavy} to \ifshort{}{each of }these prior results.

\Citet{fb_heavy_zwart} study FB in heavy traffic.
They show that if $X \in \OR{1}{2}$, then
\begin{align*}
  \E{\response{\fb{1}}{}}
  = \Theta\gp*{\ifshort{-\log(1 - \rho)}{\log\frac{1}{1 - \rho}}},
\end{align*}
matching the first expression in \cref{thm:heavy}.
They also show that if $X \in \OR{2}{\infty} \cup \Gumbel$, then
\begin{align*}
  \E{\response{\fb{1}}{}}
  = \Theta\ifshort{\gp[\Big]}{\gp*}{\ifshort{\gp[\Big]}{\frac{1}}{(1 - \rho) \cdot \rank{\serpt{}}[\big]{\Ginv{1 - \rho}}}\ifshort{\sp{-1}}{}}.
\end{align*}
This is similar to the second expression in \cref{thm:heavy},
except it replaces the monotonic $\rank{\mserpt{}}{}$
with the nonmonotonic $\rank{\serpt{}}{}$,
which pinpoints the suboptimality of FB's heavy-traffic scaling.

\Citet{srpt_heavy_zwart} study SRPT in heavy traffic.
They show that if $X \in \OR{1}{2}$, then
\begin{align*}
  \E{\response{\srpt{1}}{}}
  = \Theta\gp*{\ifshort{-\log(1 - \rho)}{\log\frac{1}{1 - \rho}}},
\end{align*}
and if $\F{}$ has upper Matuszewska index less than~$-2$,
which covers $X \in \OR{2}{\infty} \cup \Gumbel$, then
\begin{align*}
  \E{\response{\srpt{1}}{}}
  = \Theta\ifshort{\gp[\big]}{\gp*}{\ifshort{\gp[\big]}{\frac{1}}{(1 - \rho) \cdot \ol{G}_{}^{-1}(1 - \rho)}\ifshort{\sp{-1}}{}},
\end{align*}
where
\begin{align*}
  \ol{G}(x) = 1 - \frac{\E{X \1(X \leq x)}}{\E{X}} = \G{x} + \frac{x \F{x}}{\E{X}}.
\end{align*}
Recall that SRPT minimizes mean response time in the presence of job size information,
whereas Gittins does not use job size information,
so the heavy-traffic scaling of SRPT is a lower bound on that of Gittins.
By comparing the above result for SRPT
with our result for Gittins (\cref{thm:heavy}),
we learn when knowledge of job sizes yields
an asymptotic improvement in mean response time.
\begin{itemize}
\item
  For $X \in \OR{1}{2}$, meaning $X$ is heavy-tailed with infinite variance,
  the heavy-traffic scaling of Gittins matches that of SRPT.
\item
  For $X \in \OR{2}{\infty}$, meaning $X$ is heavy-tailed with finite variance,
  the heavy-traffic scaling of Gittins still matches that of SRPT.
  Specifically, we later show $\rank{\mserpt{}}{a} = \Theta(a)$ (\cref{thm:yz_bound_mserpt_or}),
  and one can also show $\ol{G}_{}^{-1}(1 - \rho) = \Theta(\Ginv{1 - \rho})$.
\item
  For $X \in \Gumbel$, meaning $X$ is not heavy-tailed,
  one can show $\rank{\mserpt{}}{a} = o(a)$ \citep{book_resnick},
  implying Gittins has worse heavy-traffic scaling than SRPT in those cases.
\end{itemize}
We see that, roughly speaking, Gittins matches the heavy-traffic scaling of SRPT
if and only if the job size distribution is heavy-tailed.
We conclude that knowledge of job sizes yields an asymptotic improvement in mean response time
for non-heavy-tailed job size distributions.

\Citet{rmlf_zwart} study RMLF in heavy traffic.
They show that if $\E{X^\alpha} < \infty$
for some $\alpha > 2$, then
\begin{align}
  \label{eq:rmlf_heavy}
  \E{\response{\rmlf{1}}{}}
  = O\ifshort{\gp}{\gp*}{\ifshort{-\log(1 - \rho) \cdot}{} \E{\response{\srpt{1}}{}} \ifshort{}{\cdot \log\frac{1}{1 - \rho}}}.
\end{align}
Because Gittins minimizes \mg{1} mean response time,
this serves as an upper bound on the heavy-traffic scaling of Gittins.
However, as previously discussed when comparing \cref{thm:heavy}
to prior results on SRPT,
there are cases where Gittins matches the heavy-traffic scaling of SRPT,
so our result is a tighter bound.
With that said,
requiring $\E{X^\alpha} < \infty$ for some $\alpha > 2$
is more lenient than the precondition of \cref{thm:heavy},
so there are still instances where
\cref{eq:rmlf_heavy} is the best known bound on Gittins's heavy-traffic scaling.

\section{Technical Overview}
\label{sec:overview}

Our main goal is to show that \mgittins{}
minimizes \mg{k} mean response time in the $\rho \to 1$ limit.
Specifically, we show
\begin{align}
  \label{eq:high_level}
  \E{\response{\mgittins{k}}{}} \le
  \E{\response{\gittins{1}}{}}
  + o(\E{\response{\gittins{1}}{}}).
\end{align}

The only existing technique for proving a bound like \cref{eq:high_level}
is the \mg{k} tagged job method of \citet{srpt_multiserver_grosof}.
In general, tagged job methods work as follows
\citep{soap_scully,
  srpt_analysis_schrage,
  fb_analysis_schrage,
  multiclass_ayesta,
  book_kleinrock,
  smart_insensitive_wierman,
  book_harchol-balter,
  srpt_multiserver_grosof}:
one focuses on a ``tagged'' job $J$ throughout its time in the system,
tracking how much each other job delays $J$.
The amount of time for which another job can delay $J$
is called the \emph{relevant work} due to that other job.
The specific \mg{k} tagged job method \citep{srpt_multiserver_grosof}
relates the amount of relevant work in an \mg{k} under \generic{k}
to the amount of relevant work in an \mg{1} under~\generic{1}.

As a first approach, we might try to prove a result like \cref{eq:high_level}
for \gittins{k} using the \mg{k} tagged job method.
Unfortunately, the method turns out not to work for \gittins{},
because \gittins{} can have a nonmonotonic rank function.
It turns out that under nonmonotonic rank functions,
jobs can contribute more relevant work in an \mg{k} than in an \mg{1}\ifshort{}{
(\cref{app:mgk_nonmonotonic})},
resulting in a much looser response time bound.

\begin{figure}
  \centering
  \input{fig_overview}
  \caption{Proof Overview}
  \label{fig:overview}
\end{figure}

Our key insight is that we can generalize the \mg{k} tagged job method
of \citet{srpt_multiserver_grosof}
to any SOAP policy, provided it has a monotonic rank function.
In \cref{thm:mgk_response} we show that for any monotonic SOAP policy~\generic{},
\begin{align}
  \label{eq:mgk_response}
  \E{\response{\generic{k}}{}}
  \le \E{\waiting{\generic{1}}{}}
  + k\E{\residence{\generic{1}}{}}
  + (k-1)\E{\sesidence{\generic{1}}{}},
\end{align}
where the quantities on the right hand side,
defined formally in \cref{sec:mgk},
can be thought of as follows:
\begin{itemize}
\item
  $\waiting{\generic{1}}{}$ and $\residence{\generic{1}}{}$
  are distributions called
  \emph{waiting time} and \emph{residence time}, respectively
  \citep{soap_scully}.
  Response time in the \mg{1} is the sum of waiting time and residence time.
\item
  $\sesidence{\generic{1}}{}$ is a new distribution we call
  \emph{inflated residence time},
  which is similar to residence time but longer.
\end{itemize}
Proving \cref{eq:mgk_response} is the first stepping stone
to proving \cref{thm:mgittins_opt}
because it reduces an \mg{k} analysis to an \mg{1} analysis.
Only the $\E{\residence{\generic{1}}{}}$
and $\E{\sesidence{\generic{1}}{}}$ coefficients depend on~$k$,
so to prove \cref{thm:mgittins_opt},
we show the $\E{\waiting{\generic{1}}{}}$ term dominates\ifshort{}{
  in the $\rho \to 1$ limit}
when \generic{} is \mgittins{}.
\Cref{fig:overview} gives an overview of the main proof steps.

In the remainder of this section,
our goal is to bound \QRS[\generic{1}],
where \generic{} is either \mgittins{} or \mserpt{}.
We begin in \cref{sub:understanding_qrs} by explaining in more detail
the concepts of relevant work and of waiting, residence, and inflated residence time.
In doing so, we introduce \emph{age cutoffs},
quantities which characterize the relevant work due to each job.
It turns out that to bound \QRS[\generic{1}],
we first need to bound the age cutoffs.
\Cref{sub:age_cutoffs} presents our age cutoff bounds,
deferring proofs to \cref{sec:rank_bounds},
and \cref{sub:heavy_traffic_qrs} presents our bounds on \QRS[\generic{1}],
deferring proofs to \cref{sec:heavy}.
Finally, in \cref{sub:formal_proofs},
we formally prove \cref{thm:mgittins_opt, thm:mserpt_approx, thm:heavy}
by combining the intermediate results discussed throughout this section.

\subsection{Understanding the Tagged Job Method and Relevant Work}
\label{sub:understanding_qrs}

In this section we give intuition for the tagged job method,
deferring some formalities to \cref{sec:mgk}.

Recall that the tagged job method works by focusing on the journey of
a ``tagged'' job~$J$ through the system.
Roughly speaking, the relevant work due to any other job is
the amount of time by which that job delays $J$'s departure.
A key insight from the \mg{1} SOAP analysis \citep{soap_scully}
is that to figure out how much another job delays~$J$,
we need to look not at $J$'s current rank
but at its \emph{worst future rank}.
This is because even if $J$ has priority over another job at first,
if $J$'s rank later increases, the other job can get priority.

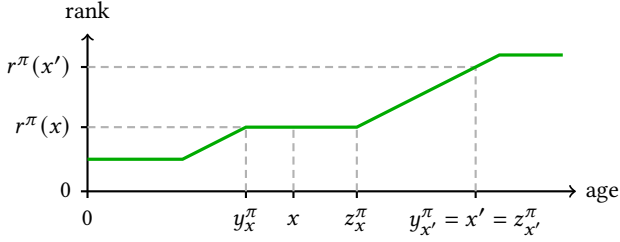
\begin{figure}
  \centering
  \input{fig_yz}
  \captionsqueeze
  \caption{New Job and Old Job Age Cutoffs}
  \label{fig:yz}
\end{figure}

Suppose that $J$ has size~$x$.
Under a monotonic SOAP policy~\generic{},
such as \mgittins{} or \mserpt{},
the worst future rank $J$ will have
is always the rank it will have just before completion, namely $\rank{\generic{}}{x}$.
The amount of relevant work due to another job~$J'$
is the amount of time $J'$ is served while $J$ is in the system
until $J'$ either completes or reaches rank $\rank{\generic{}}{x}$.
Due to the FCFS tiebreaking rule (\cref{sub:rank_functions}),
exactly what ``reaches'' means depends on when $J'$ arrives.
\begin{itemize}
\item
  \emph{New jobs}, those that arrive \emph{after} $J$,
  contribute relevant work until
  they first have rank greater than \emph{or equal} to $\rank{\generic{}}{x}$.
  This occurs at a specific age called the \emph{new job age cutoff},
  denoted~$\y[\generic{}]$.
\item
  \emph{Old jobs}, those that arrive \emph{before} $J$,
  contribute relevant work until
  they first have rank \emph{strictly} greater than $\rank{\generic{}}{x}$.
  This occurs at a specific age called the \emph{old job age cutoff},
  denoted~$\z[\generic{}]$.
\end{itemize}
\Cref{fig:yz} illustrates the new job and old job age cutoffs
$\y[\generic{}]$ and~$\z[\generic{}]$,
which are formally defined below.\footnote{%
  The new job and old job age cutoffs of~$x$ are equivalent to
  what \citet{m-serpt_scully} call the \emph{previous and next hill ages} of~$x$.}
Roughly speaking,
\begin{itemize}
\item
  if $\rank{\generic{}}{}$ is increasing at~$x$,
  then $\y[\generic{}] = x = \z[\generic{}]$;
  and
\item
  if $\rank{\generic{}}{}$ is constant at~$x$,
  then $\y[\generic{}]$ and $\z[\generic{}]$ are
  the endpoints of the constant region containing~$x$.
\end{itemize}
As \cref{fig:yz} illustrates, we always have
\begin{align}
  \label{eq:yxz}
  \y[\generic{}] \leq x \leq \z[\generic{}].
\end{align}

\begin{definition}
  \label{def:yz}
  Let \generic{} be a monotonic SOAP policy.
  The \emph{new job age cutoff} and \emph{old job age cutoff} of size~$x$
  are, respectively,
  \begin{align*}
    \y[\generic{}]
    &= \sup\curlgp{a \geq 0 \given
        \rank{\generic{}}{a} < \rank{\generic{}}{x}}, \\
    \z[\generic{}]
    &= \sup\curlgp{a \geq 0 \given
        \rank{\generic{}}{a} \leq \rank{\generic{}}{x}}.
  \end{align*}
  When the policy in question is clear,
  we drop the superscript~\generic{}.
\end{definition}

One can use new job and old job age cutoffs
to write \mg{1} mean response time under a monotonic SOAP policy
\citep{m-serpt_scully}.
As a first step, we write \mg{1} response time~$\response{\generic{1}}{}$
as a sum of two parts,
called \emph{waiting time}~$\waiting{\generic{1}}{}$
and \emph{residence time}~$\residence{\generic{1}}{}$ \citep{soap_scully}:
\begin{align*}
  \E{\response{\generic{1}}{}}
  = \E{\waiting{\generic{1}}{}} + \E{\residence{\generic{1}}{}}.
\end{align*}
We define waiting and residence times formally in \cref{sec:mgk}.
For now, we just need to know that
their means can be written in terms of $\y[\generic{}]$ and~$\z[\generic{}]$.
Specifically, \citet[Propositions~4.7 and~4.8]{m-serpt_scully} show
\begin{align}[c]
  \label{eq:waiting_residence_monotonic}
  \E{\waiting{\generic{1}}{}}
  &= \int_0^\infty
      \frac{\excess{\z[\generic{}]}}{\coload{\y[\generic{}]} \coload{\z[\generic{}]}}
    \dF{x},
  \\
  \E{\residence{\generic{1}}{}}
  &= \int_0^\infty \frac{x}{\coload{\y[\generic{}]}} \dF{x},
\end{align}
where $\coload{}$ and $\excess{}$ are defined as
\begin{align}[c]
  \label{eq:coload_excess}
  \coload{a} &= 1 -  \lambda \E{\min\{X, a\}} = 1 - \int_0^a \lambda \F{t} \d{t}, \\
  \excess{a} &= \frac{\lambda}{2} \E{\min\{X, a\}^2} = \int_0^a \lambda t\F{t} \d{t}.
\end{align}
The proof of \cref{thm:mgk_response}
explains the intuition behind \cref{eq:waiting_residence_monotonic}.

The significance of \cref{eq:mgk_response}
is that it expresses \mg{k} response time
in terms of waiting and residence times,
which are \mg{1} quantities.
It also features a third quantity
called \emph{inflated residence time}~$\sesidence{\generic{1}}{}$.
We define inflated residence time formally in \cref{sec:mgk}.
For now, we just need to know that its mean,
\begin{align}
  \label{eq:sesidence}
  \E{\sesidence{\generic{1}}{}}
  = \int_0^\infty \frac{\z[\generic{}]}{\coload{\y[\generic{}]}} \dF{x},
\end{align}
can be written in terms of $\y[\generic{}]$ and~$\z[\generic{}]$.
Note that $\E{\residence{\generic{1}}{}} \leq \E{\sesidence{\generic{1}}{}}$.

\subsection{Bounding New and Old Age Cutoffs}
\label{sub:age_cutoffs}

\begin{table}
  \ifshort{\shortenRefs}{}
  \centering
  \begin{threeparttable}
    \caption{New Job and Old Job Age Cutoff Bounds}
    \label{tab:yz_bounds}
    \input{tab_yz}
    \onetablenote
    \begin{tablenotes}
    \item
      These bounds on $\y[\generic{}]$ and $\z[\generic{}]$
      are critical for characterizing heavy-traffic scaling
      of \QRS[\generic{1}].
    \end{tablenotes}
  \end{threeparttable}
\end{table}

Recall that proving our main results rests on
characterizing the heavy-traffic scaling of \QRS[\generic{}],
where \generic{} is either \mgittins{} or \mserpt{}.
As we see in \cref{eq:waiting_residence_monotonic, eq:sesidence},
both $\y[\generic{}]$ and $\z[\generic{}]$ feature prominently in
the formulas of \QRS[\generic{}].
This means the first step of characterizing the heavy-traffic scaling
of \QRS[\generic{}]
is understanding $\y[\generic{}]$ and~$\z[\generic{}]$.
This is the subject of \cref{sec:rank_bounds}, in which
we prove bounds on $\y[\generic{}]$ and $\z[\generic{}]$
for a wide class of job size distributions.
\Cref{tab:yz_bounds} summarizes these results.
The main takeaway is that $\y[\generic{}]$ and $\z[\generic{}]$
are always polynomially bounded relative to~$x$.

\subsection{Characterizing Heavy Traffic Scaling}
\label{sub:heavy_traffic_qrs}

Armed with bounds on age cutoffs,
we are ready to characterize heavy-traffic scaling of
mean waiting, residence, and inflated residence times.
This is the subject of \cref{sec:heavy}, in which\ifshort{
  \cref{thm:heavy_mserpt_or_iv, thm:heavy_mserpt_or_fv,
    thm:heavy_mserpt_gumbel, thm:heavy_mserpt_enbue}
  characterize \mserpt{}'s heavy-traffic scaling; and
  \cref{thm:heavy_mgittins_response, thm:heavy_mgittins_sesidence}
  characterize \mgittins{}'s heavy-traffic scaling in terms of \mserpt{}'s.%
}{
\begin{itemize}
\item
  \cref{thm:heavy_mserpt_or_iv, thm:heavy_mserpt_or_fv,
    thm:heavy_mserpt_gumbel, thm:heavy_mserpt_enbue}
  characterize \mserpt{}'s heavy-traffic scaling; and
\item
  \cref{thm:heavy_mgittins_response, thm:heavy_mgittins_sesidence}
  characterize \mgittins{}'s heavy-traffic scaling in terms of \mserpt{}'s.
\end{itemize}%
}
\Cref{tab:heavy} summarizes these results.
The main takeaway of the table is that
for all of the finite-variance job size distribution classes considered,\footnote{%
  That is, for all the classes in \cref{tab:heavy} except $\OR{1}{2}$.}
if \generic{} is either \mgittins{} or \mserpt{},
$\E{\waiting{\generic{1}}{}}$ dominates
$\E{\residence{\generic{1}}{}}$ and $\E{\sesidence{\generic{1}}{}}$,
with the latter sometimes requiring an additional condition.
Specifically,
\begin{itemize}
\item
  $\E{\waiting{\generic{1}}{}}$ grows polynomially in $1/(1 - \rho)$, whereas
\item
  $\E{\residence{\generic{1}}{}}$ and $\E{\sesidence{\generic{1}}{}}$
  grow subpolynomially in $1/(1 - \rho)$.
\end{itemize}

\ifshort{\begin{table}}{\begin{table*}}
  \ifshort{\shortenRefs\small}{}
  \centering
  \begin{threeparttable}
    \caption{Heavy-Traffic Scaling of Waiting, Residence, and \ifshort{\\}{}Inflated Residence Times}
    \label{tab:heavy}
    \input{tab_heavy}
    \onetablenote
    \begin{tablenotes}
    \item
      These bounds hold when \generic{} is either \mgittins{} or \mserpt{}\ifshort{.
      Additionally, $\E{\sesidence{\mserpt{1}}{}} = O((1 - \rho)^{-\epsilon})$ for all $\epsilon > 0$
      for the $\Gumbel \cap \QIMRL$ class.}{,
      except for the $\Gumbel \cap \QIMRL$ case,
      in which the bound holds only for \mserpt{}.}
    \end{tablenotes}
  \end{threeparttable}
\ifshort{\end{table}}{\end{table*}}

\subsection{From Intermediate Results to Main Results}
\label{sub:formal_proofs}

We now prove our main results.
The proofs of \cref{thm:mgittins_opt, thm:mserpt_approx}
both follow the same three \ifshort{}{main }steps,
where \generic{} is \mgittins{} or \mserpt{}, respectively:\ifshort{
  \cref{thm:mgk_response}
  bounds $\E{\response{\generic{k}}{}}$ in terms of \mg{1} quantities,
  the results in \cref{tab:heavy} show
  $\lim_{\rho \to 1} \E{\response{\generic{k}}{}} / \E{\waiting{\generic{1}}{}} = 1$,
  and prior work relates $\E{\waiting{\generic{1}}{}}$
  to $\E{\response{\gittins{1}}{}}$.
}{
\begin{itemize}
\item
  \Cref{thm:mgk_response}
  bounds $\E{\response{\generic{k}}{}}$ in terms of \mg{1} quantities.
\item
  The results in \cref{tab:heavy} show
  $\lim_{\rho \to 1} \E{\response{\generic{k}}{}} / \E{\waiting{\generic{1}}{}} = 1$.
\item
  Prior work relates $\E{\waiting{\generic{1}}{}}$
  to $\E{\response{\gittins{1}}{}}$.
\end{itemize}}

\begin{proof}[Proof of \cref{thm:mgittins_opt}]
  An \mg{1} can simulate any \mg{k} policy by sharing the server,
  so the fact that Gittins minimizes \mg{1} mean response time
  means $\E{\response{\mgittins{k}}{}} / \E{\response{\gittins{1}}{}} \geq 1$.
  It therefore suffices to show
  $\lim_{\rho \to 1} \E{\response{\mgittins{k}}{}} / \E{\response{\gittins{1}}{}} \leq 1$.

  \Cref{thm:mgk_response} implies
  \begin{align*}
    \frac{\E{\response{\mgittins{k}}{}}}{\E{\waiting{\mgittins{1}}{}}}
    \leq 1 + \frac{
      k\E{\residence{\mgittins{1}}{}} + (k - 1) \E{\sesidence{\mgittins{1}}{}}
    }{\E{\waiting{\mgittins{1}}{}}}.
  \end{align*}
  \Cref{thm:heavy_mserpt_or_fv, thm:heavy_mserpt_gumbel, thm:heavy_mserpt_enbue,
    thm:heavy_mgittins_response, thm:heavy_mgittins_sesidence}
  imply that the second term vanishes in the $\rho \to 1$ limit.
  A result of \citet[Proposition~4.7]{m-serpt_scully} implies
  \begin{align}
    \label{eq:mgittins_gittins}
    \E{\waiting{\mgittins{1}}{}}
    \leq \E{\waiting{\gittins{1}}{}}
    \leq \E{\response{\gittins{1}}{}},
  \end{align}
  implying the desired result.
\end{proof}

\begin{proof}[Proof of \cref{thm:mserpt_approx}]
  \Cref{thm:mgk_response} implies
  \begin{align*}
    \frac{\E{\response{\mserpt{k}}{}}}{\E{\waiting{\mserpt{1}}{}}}
    \leq 1 + \frac{
      k\E{\residence{\mserpt{1}}{}} + (k - 1) \E{\sesidence{\mserpt{1}}{}}
    }{\E{\waiting{\mserpt{1}}{}}}.
  \end{align*}
  \Cref{thm:heavy_mserpt_or_fv, thm:heavy_mserpt_gumbel, thm:heavy_mserpt_enbue}
  imply that the second term vanishes in the $\rho \to 1$ limit.
  \Citet[Lemma~5.6]{m-serpt_scully} show\ifshort{ that}{}\footnote{%
    While \citet[Lemma~5.6]{m-serpt_scully}
    mention Gittins instead of \mgittins{},
    they prove the desired statement for \mgittins{}
    as an intermediate step of their proof.}
  \ifshort{$\E{\waiting{\mserpt{1}}{}} \leq 2\E{\waiting{\mgittins{1}}{}},$}{
  \begin{align*}
    \E{\waiting{\mserpt{1}}{}} \leq 2\E{\waiting{\mgittins{1}}{}},
  \end{align*}}
  which combines with \cref{eq:mgittins_gittins}
  to imply the desired result.
\end{proof}

To prove \cref{thm:heavy}, we simply combine the results in \cref{tab:heavy}.

\begin{proof}[Proof of \cref{thm:heavy}]
  We examine each case in turn.
  \begin{itemize}
  \item
    For $X \in \OR{1}{2}$, we use
    \cref{thm:heavy_mserpt_or_iv, thm:heavy_mgittins_response}.
  \item
    For $X \in \OR{2}{\infty} \cup \Gumbel$,
    we use
    \cref{thm:heavy_mserpt_or_fv, thm:heavy_mserpt_gumbel,
      thm:heavy_mgittins_response}.
  \item
    For $X \in \ENBUE$, we have $\rank{\mserpt{}}{a} = \Theta(1)$ by \cref{def:enbue},
    so we use
    \cref{thm:heavy_mserpt_enbue, thm:heavy_mgittins_response}.
    \qedhere
  \end{itemize}
\end{proof}

\section{\mg{k} Response Time Bound}
\label{sec:mgk}

\begin{table}
  \ifshort{\shortenRefs}{}
  \centering
  \begin{threeparttable}
    \caption{Summary of Notation}
    \label{tab:notation}
    \input{tab_notation}
    \onetablenote
    \begin{tablenotes}
    \item
      Additionally, $\response{\generic{k}}{x}$
      is size-conditional response time for size~$x$, and similarly for
      $\waiting{\generic{1}}{x}$, $\residence{\generic{1}}{x}$,
      and $\sesidence{\generic{1}}{x}$.
    \end{tablenotes}
  \end{threeparttable}
\end{table}

This section bounds \mg{k} mean response time
under any monotonic SOAP policy~\generic{}.
The notation used in \cref{thm:mgk_response} below
is summarized in \cref{tab:notation}.

\begin{theorem}
  \label{thm:mgk_response}
  For any monotonic SOAP policy \generic{},
  \begin{align}
    \label{eq:response_detailed}
    \E{\response{\generic{k}}{x}}
    \leq \ifshort{}{\frac{1}{\coload{\y[\generic{}]}}} \gp*{
        \ifshort{\slashfrac}{\frac}{\excess{\z[\generic{}]}}{\coload{\z[\generic{}]}}
        + kx
        + (k-1)\z[\generic{}]
    }\ifshort{/\coload{\y[\generic{}]}}{},
  \end{align}
  and therefore
  \ifshort{\(}{\begin{align*}}
    \E{\response{\generic{k}{}}{}}
    \leq \E{\waiting{\generic{1}}{}}
      + k \E{\residence{\generic{1}}{}}
      + (k - 1) \E{\sesidence{\generic{1}}{}}\ifshort{}{.}
  \ifshort{\).}{\end{align*}}
\end{theorem}

\begin{proof}
  In order to bound \mg{k} mean response time,
  we use a tagged job method in the style of \citet{srpt_multiserver_grosof},
  but we generalize it to allow an arbitrary monotonic SOAP policy~\generic{}.
  We consider an arbitrary ``tagged'' job $J$ of size~$x$
  arriving to a steady-state system.
  Our goal is to analyze the distribution of $J$'s response time.

  The first step is a shift in perspective:
  instead of thinking about \emph{time passing},
  we reason in terms of \emph{work completed}.
  Since each of the $k$ servers works at rate~$1/k$,
  the system can complete work at rate~$1$.
  While $J$ is in the system,
  servers sometimes complete work and are sometimes left idle\ifshort{, so}{.
  This means} $J$'s response time is the sum of
  \begin{itemize}
  \item
    the amount of work completed while $J$ is in the system and
  \item
    the amount of work ``wasted'', meaning service capacity left idle,
    while $J$ is in the system.
  \end{itemize}
  We bound $J$'s response time by bounding the total amount of work above.
  We do so by dividing it into several pieces:
  \begin{itemize}
  \item \emph{Tagged work}: the work of $J$ itself.
  \item \emph{Virtual work}: work on jobs prioritized behind $J$,
    plus wasted work due to servers left idle.
  \item \emph{Relevant work}: work on jobs prioritized ahead of~$J$.
    We divide this into two subcategories:
    \begin{itemize}
    \item \emph{Old} relevant work: relevant work on \emph{old jobs},
      namely those present when $J$~arrives.
    \item \emph{New} relevant work: relevant work on \emph{new jobs},
      namely those that arrive after~$J$.
    \end{itemize}
  \end{itemize}

  For the first two categories,
  we have the same simple bound as \citet{srpt_multiserver_grosof}:
  tagged work and virtual work add up to at most~$kx$.
  This is because tagged work is $J$'s size~$x$,
  and the scheduling policy ensures that
  a server only completes virtual work while $J$ is in service at another server.
  However, bounding the two relevant work categories
  is more complicated than in \citet{srpt_multiserver_grosof}.

  We begin by asking: what rank must a job have to contribute to relevant work?
  Note that the job $J$ will never have rank greater than its rank upon completion,
  $\rank{\generic{}}{x}$,
  since \generic{} is a monotonic policy.
  As a result, all new relevant work is from jobs with rank
  \emph{strictly} less than $\rank{\generic{}}{x}$,
  and all old relevant work is from jobs with rank
  less than \emph{or equal} to $\rank{\generic{}}{x}$.
  We can put this in terms of the age cutoffs defined in \cref{def:yz}:
  \begin{itemize}
  \item
    jobs contribute new relevant work up to at most age~$\y[\generic{}]$, and
  \item
    jobs contribute old relevant work up to at most age~$\z[\generic{}]$.
  \end{itemize}
  In the rest of this proof,
  $\y$ and $\z$ refer to $\y[\generic{}]$ and~$\z[\generic{}]$, respectively.

  To help us bound the amount of old relevant work completed while $J$
  is in the system,
  we define a new concept:
  the amount of relevant work in the \mg{k} system under~\generic{}.

  \begin{definition}
    Let $\relwork{\generic{k}}{x}(t)$ denote the amount of work
    in the \mg{k} at time $t$ which is relevant to a job~$J$ of size~$x$:
    \begin{align*}
      \relwork{\generic{k}}{x}(t)
        = \smashoperator{\sum_{\text{jobs } J'}} \gp[\big]{\min\{\z, x_{J'}\} - a_{J'}(t)}^+,
    \end{align*}
    where $x_{J'}$ is the size of job~$J'$ and $a_{J'}(t)$ is its age at time~$t$.
    We write $\relwork{\generic{k}}{x}$ for the steady state distribution
    of the amount of relevant work in the \mg{k} system.
  \end{definition}

  Since $J$ is a Poisson arrival,
  $\relwork{\generic{k}}{x}$ is the distribution of the amount of relevant work
  in the system when $J$ arrives.
  That amount is an upper bound on the amount of old relevant work
  that will be completed while $J$ is in the system.

  To bound new relevant work,
  note that if a job~$J'$ of size $x'$ arrives while $J$ is in the system,
  then $J'$ contributes at most $\min\{x', y_x\}$ new relevant work.
  As a result, new relevant work can be upper bounded
  by considering a transformed \mg{1} system in which the job size distribution is
  \ifshort{$\min\{X, \y\}$.

  }{%
  \begin{align*}
    X_{\ol{\y}} =_\st \min\{X, \y\}.
  \end{align*}%
  }
  The amount of new relevant work that arrives to our real system is
  upper bounded by the total amount of work that arrives to the transformed system.
  Let $B_{\ol{\y}}(w)$ be the length of a busy period in the transformed \mg{1} system
  started by an initial amount of work~$w$.
  If $w$ is the total amount of tagged, virtual, and old relevant work,
  then the amount of new relevant work is at most $B_{\ol{\y}}(w) - w$.

  Combining our bounds, we obtain
  \begin{align*}
    \response{\generic{k}}{x} \le_\st B_{\ol{\y}}\gp[\big]{kx + \relwork{\generic{k}}{x}}.
  \end{align*}
  Applying \Cref{lem:delta_relwork}, stated\ifshort{}{ and proven} later in this section,
  yields
  \begin{align}
    \label{eq:response_bound_busy}
    \response{\generic{k}}{x} \le_\st B_{\ol{\y}}\gp[\big]{kx + \relwork{\generic{1}}{x} + (k-1)\z}.
  \end{align}
  Taking expectations gives us
  \begin{align*}
    \E{\response{\generic{k}}{x}}
    \leq \ifshort{\slashfrac}{\frac}{
      \ifshort{\gp[\big]}{}{
        \E{\relwork{\generic{1}}{x}}
        + kx
        + (k-1)\z
      }}{\coload{\y}}.
  \end{align*}
  Because \generic{1} is work conserving with respect to relevant work,
  the Pollaczek-Khinchine formula tells us
  \ifshort{\(}{\begin{align*}}
    \E{\relwork{\generic{1}}{x}} = \ifshort{\slashfrac}{\frac}{\excess{\z}}{\coload{\z}}\ifshort{}{,}
  \ifshort{\),}{\end{align*}}
  which completes the proof of~\cref{eq:response_detailed}.

  To connect \cref{eq:response_detailed} to the quantities \QRS[\generic{}],
  we rewrite \cref{eq:response_bound_busy} as
  \begin{align}
    \label{eq:response_bound_split}
    \ifshort{\hspace{-1em}}{}
    \response{\generic{k}}{x}
    \le_\st B_{\ol{\y}}(\relwork{\generic{1}}{x})
    + \sum\ifshort{\nolimits}{}_{1}^k B_{\ol{\y}}(x)
    + \sum\ifshort{\nolimits}{}_{1}^{k-1} B_{\ol{\y}}(\z),
  \end{align}
  where all of the relevant busy periods are independent.
  Prior work on SOAP policies \citep{soap_scully, m-serpt_scully}
  gives names to some of the distributions on the right-hand side.\ifshort{}{\footnote{%
    We define waiting, residence, and inflated residence times
    in terms of relevant busy periods.
    Waiting and residence times also have natural definitions as components of
    \mg{1} response time \citep{soap_scully, m-serpt_scully},
    but we do not need them in this paper.}}
  \begin{itemize}
  \item
    The \emph{size-conditional waiting time} for size~$x$
    is the random variable
    $\waiting{\generic{1}}{x} =_\st B_{\ol{\y}}(\relwork{\generic{1}}{x})$,
    and \emph{waiting time} is
    $\waiting{\generic{1}}{} =_\st \waiting{\generic{1}}{X}$.
  \item
    The \emph{size-conditional residence time} for size~$x$
    is the random variable
    $\residence{\generic{1}}{x} =_\st B_{\ol{\y}}(x)$,
    and \emph{residence time} is
    $\residence{\generic{1}}{} =_\st \residence{\generic{1}}{X}$.
  \item
    As there is no concise name for $B_{\ol{\y}}(\z)$ in prior work,
    \ifshort{let}{we define} \emph{size-conditional inflated residence time} for size~$x$
    to be the random variable
    $\sesidence{\generic{1}}{x} =_\st B_{\ol{\y}}(\z)$,
    and \ifshort{let}{we define} \emph{inflated residence time} \ifshort{}{to }be
    $\sesidence{\generic{1}}{} =_\st \sesidence{\generic{1}}{X}$.
  \end{itemize}
  With these definitions in place, \cref{eq:response_bound_split} gives us
  \begin{align*}
    \response{\generic{k}}{x}
    \le_\st \waiting{\generic{1}}{x}
    + \sum\ifshort{\nolimits}{}_{1}^k \residence{\generic{1}}{x}
    + \sum\ifshort{\nolimits}{}_{1}^{k-1} \sesidence{\generic{1}}{x},
  \end{align*}
  so the result follows by taking the expectation of
  $\response{\generic{k}}{} =_\st \response{\generic{k}}{X}$.
\end{proof}

\Cref{thm:mgk_response} applies only to monotonic SOAP policies.
It is tempting to try to apply the same technique
to SOAP policies with nonmonotonic rank functions,
but \ifshort{it}{as we discuss in \cref{app:mgk_nonmonotonic},
the argument} does not readily generalize.

\ifshort{%
  The proof of \cref{thm:mgk_response} assumes a bound on $\relwork{\generic{k}}{x}$,
  namely \cref{lem:delta_relwork}.
  The lemma generalizes a similar result of
  \citet[Lemma~7.10]{srpt_multiserver_grosof},
  so we omit its proof for lack of space.
}{%
  The proof of \cref{thm:mgk_response} assumes a bound on $\relwork{\generic{k}}{x}$.
  We prove the bound in the following lemma,
  which generalizes a similar lemma of
  \citet[Lemma~7.10]{srpt_multiserver_grosof}.%
}

\begin{lemma}
  \label{lem:delta_relwork}
  Let
  \ifshort{\(}{\begin{align*}}
    \Delta_x(t) = \relwork{\generic{k}}{x}(t) - \relwork{\generic{1}}{x}(t)\ifshort{}{.}
  \ifshort{\).}{\end{align*}}
  Then $\Delta_x(t) \leq (k-1)\z[\generic{}]\esub{}$ for all times~$t$,
  and therefore
  \begin{align*}
    \relwork{\generic{k}}{x} \le_\st \relwork{\generic{1}}{x} + (k-1)\z[\generic{}].
  \end{align*}
\end{lemma}

\ifshort{}{\begin{proof}
  Throughout this proof, $\z$ refers to $\z[\generic{}]$.
  We consider a pair of coupled systems with the same arrival sequence:
  \begin{itemize}
  \item \emph{System~1}, an \mg{1} using \generic{1}; and
  \item \emph{System~$k$}, an \mg{k} using \generic{k}.
  \end{itemize}
  Our approach is to bound the difference in relevant work between
  Systems~1 and~$k$ at any time~$t$.

  Call a job \emph{relevant} if it has age less than~$\z$.
  These are the only jobs that contribute relevant work.
  To bound $\Delta_x(t)$, we divide times~$t$ into two types of intervals:
  \begin{itemize}
  \item \emph{few-jobs intervals}, during which there are fewer than
    $k$ relevant jobs in System~$k$; and
  \item \emph{many-jobs intervals}, during which there are at least
    $k$ relevant jobs in System~$k$.
  \end{itemize}
  Note that both types of intervals are defined based on System~$k$ alone,
  so System~1 may or may not have relevant jobs during either type of interval.

  Any time $t$ is in either a few-jobs interval or a many-jobs interval.
  If $t$ is in a few-jobs interval, the argument is simple:
  there are at most $k-1$ relevant jobs in System~$k$ at time $t$, so
  \begin{align*}
    \Delta_x(t) \le \relwork{\generic{k}}{x}(t) \le (k-1)\z.
  \end{align*}
  Suppose instead that $t$ is in a many-jobs interval.
  Let $s \leq t$ be the start of the many-jobs interval containing~$t$.
  We will show
  \begin{align*}
    \Delta_x(t) \le \Delta_x(s) \le (k-1)\z.
  \end{align*}

  We begin by showing $\Delta_x(t) \le \Delta_x(s)$.
  Note that arrivals do not affect~$\Delta_x$,
  because the two systems experience the same arrivals
  and have the same definition of relevant work.
  Next, note that service to irrelevant jobs
  does not affect~$\Delta_x$,
  because irrelevant jobs never become relevant under~\generic{},
  since \generic{} is a monotonic policy.
  In fact, the only way that $\Delta_x$ changes over a many-jobs period
  is due to service to relevant jobs.
  System~$k$ serves relevant jobs on all $k$ servers
  throughout a many-jobs period,
  completing relevant work at rate~$1$.
  System~1 may or may not serve relevant jobs
  during a many-jobs period,
  so it completes relevant work at rate at most~$1$.
  This means $\Delta_x(t) \le \Delta_x(s)$, as desired.

  All that remains is to show that $\Delta_x(s) \le (k-1)\z$.
  Recall that $s$ is the start of a many-jobs interval.
  Many-jobs intervals cannot start due to irrelevant jobs becoming relevant,
  because \generic{} is a monotonic policy.
  This means each many-jobs interval starts due to a relevant job arriving
  while System~$k$ has $k-1$ relevant jobs.
  Relevant jobs arriving do not change~$\Delta_x$, as discussed above.
  This means $\Delta_x(s) = \Delta_x(s^-)$,
  where $s^-$ is the instant before the arrival that starts the many-jobs interval.
  But $s^-$ is in a few-jobs interval, so
  \begin{align*}
    \Delta_x(s) = \Delta_x(s^-) \le (k-1)\z.
    \qedhere
  \end{align*}
\end{proof}}

\section{Rank Function Bounds}
\label{sec:rank_bounds}

We now have a bound on \mg{k} mean response time
under monotonic SOAP policies~\generic{},
including \mgittins{} and \mserpt{}.
The bound (\cref{thm:mgk_response}) is expressed in terms of \QRS[\generic{1}],
quantities which in turn are expressed in terms of the new job and old job age cutoffs
$\y[\generic{}]$ and~$\z[\generic{}]$.
In order to prove optimality of \mgittins{} in the heavy-traffic \mg{k},
we need to understand the heavy-traffic behavior of \QRS[\generic{1}],
which, as we will see in \cref{sec:heavy},
boils down to understanding the behavior of $\y[\generic{}]$ and $\z[\generic{}]$
in the $x \to \infty$ limit.
This section is thus devoted to asymptotically bounding
the new job and old job age cutoffs,
and more generally the rank functions,
of \mgittins{} and \mserpt{}.

Recall from \cref{def:mserpt} that
SERPT's rank function is used to define \mserpt{}'s.
The following lemma shows that the two rank functions
are equal at the new job and old job age cutoffs,
and similarly for Gittins and \mgittins{}.\ifshort{}{
\Cref{fig:mserpt_serpt} gives an intuitive picture of the result.}

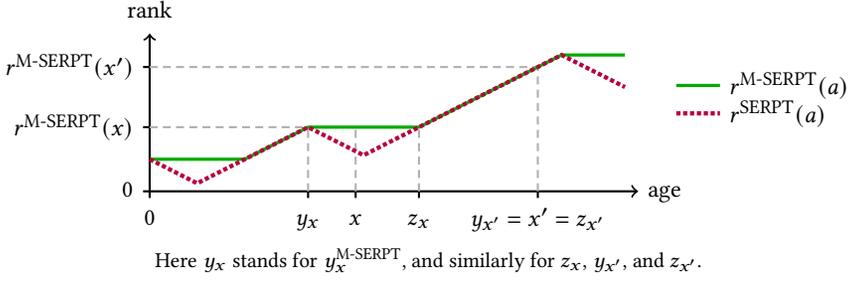
\begin{figure}
  \centering
  \input{fig_mserpt_serpt}
  \captionsqueeze
  \caption{Relationship Between \serpt{} and \mserpt{} Rank Functions}
  \label{fig:mserpt_serpt}
\end{figure}

\begin{lemma}
  \label{lem:rank_yz}
  The SERPT and \mserpt{} rank functions \ifshort{satisfy}{are related by}
  \begin{align*}
    \rank{\serpt{}}{\y[\mserpt{}]}
    \iftwocol{&}{}
    = \rank{\mserpt{}}{\y[\mserpt{}]}
    \iftwocol{\\*}{}
    = \rank{\mserpt{}}{x}
    \iftwocol{&}{}
    = \rank{\mserpt{}}{\z[\mserpt{}]}
    = \rank{\serpt{}}{\z[\mserpt{}]},
  \end{align*}
  and analogously for Gittins and \mgittins{}.
\end{lemma}

\begin{proof}
  We prove the statement for SERPT and \mserpt{},
  as the proof for Gittins and \mgittins{} is analogous.
  Throughout this proof,
  $\y$ and $\z$ refer to $\y[\mserpt{}]$ and~$\z[\mserpt{}]$, respectively.\ifshort{}{
  The illustration in \cref{fig:mserpt_serpt}
  may provide helpful intuition for the following argument.}

  We first show the outer equalities.
  \Cref{def:yz} implies that
  $\rank{\mserpt{}}{}$ is increasing in the intervals
  $(\y - \delta, \y)$ and $(\z, \z + \delta)$ for some $\delta > 0$.
  By \cref{def:mserpt}, for $\rank{\mserpt{}}{}$ to be increasing at age~$a$,
  we must have $\rank{\mserpt{}}{a} = \rank{\serpt{}}{a}$,
  so continuity of $\rank{\mserpt{}}{}$ (\ifshort{\cref{sub:rank_functions}}{\cref{lem:rank_continuous}})
  implies the outer equalities.

  By \cref{eq:yxz} and the monotonicity of~$\rank{\mserpt{}}{}$,
  it remains only to show $\rank{\mserpt{}}{\y} = \rank{\mserpt{}}{\z}$.
  This is immediate if $\y = \z$, and if $\y < \z$,
  then $\rank{\mserpt{}}{}$ is constant over the interval $[\y, \z)$,
  so the result follows by the continuity of $\rank{\mserpt{}}{}$
  (\ifshort{\cref{sub:rank_functions}}{\cref{lem:rank_continuous}}).
\end{proof}

\subsection{Bounds on the \mserpt{} Rank Function}
\label{sub:rank_bounds_mserpt}

In this section we show two bounds
on $\y[\mserpt{}]$ and $\z[\mserpt{}]$,
each subject to a different assumption on the job size distribution.

\begin{theorem}
  \label{thm:yz_bound_mserpt_or}
  If $X \in \OR{1}{\infty}$, then
  \begin{align*}
    \rank{\serpt{}}{a} &= \Theta(a), \ifshort{&}{\\}
    \ifshort{\y[\mserpt{}] &= \Theta(x), \\}{}
    \rank{\mserpt{}}{a} &= \Theta(a), \ifshort{&}{\\}
    \ifshort{}{\y[\mserpt{}] &= \Theta(x), \\}
    \z[\mserpt{}] &= \Theta(x).
  \end{align*}
\end{theorem}

\begin{proof}
  By \cref{def:or}, there exists $\alpha > 1$ such that
  \begin{align*}
    \rank{\serpt{}}{a}
    = \int_a^\infty \frac{\F{t}}{\F{a}} \d{t}
    \leq O(1)\int_a^\infty \gp*{\frac{t}{a}}^{-\alpha} \d{t}
    = O(a),
  \end{align*}
  and $\rank{\serpt{}}{a} = \Omega(a)$ follows similarly.
  This implies
  \begin{align*}
    \rank{\mserpt{}}{a}
    = \max_{b \in [0, a]} \rank{\serpt{}}{b}
    = \max_{b \in [0, a]} \Theta(b)
    = \Theta(a),
  \end{align*}
  so the result follows from \cref{lem:rank_yz}.
\end{proof}

\begin{theorem}
  \label{thm:yz_bound_mserpt_qimrl}
  If $X \in \QDHR \cup \QIMRL$ with exponent~$\gamma$, then
  \begin{align*}
    \y[\mserpt{}] &= \Omega(x^{1/\gamma}), \ifshort{&}{\\}
    \z[\mserpt{}] &= O(x^\gamma).
  \end{align*}
\end{theorem}

\begin{proof}
  The $\QDHR$ case follows from
  \cref{thm:yz_bound_mgittins_qdhr} (\cref{sub:rank_bounds_mgittins})
  and a result of \citet[Eq.~(3.8)]{m-serpt_scully} stating
  \begin{align*}
    \y[\mgittins{}] \leq \y[\mserpt{}] \leq \z[\mserpt{}] \leq \z[\mgittins{}],
  \end{align*}
  so only the $\QIMRL$ case remains.

  In the rest of this proof,
  $\y$ and $\z$ refer to $\y[\mserpt{}]$ and~$\z[\mserpt{}]$, respectively.
  By \cref{eq:yxz}, it suffices to show $\z = O(\y^\gamma)$.
  Because $X \in \QIMRL$ with exponent~$\gamma$,
  there exists strictly increasing function $m : \R_+ \to \R_+$
  such that for all ages~$a$,
  \begin{align*}
    a \leq m^{-1}\gp[\big]{\rank{\serpt{}}{a}} \leq O(a^\gamma).
  \end{align*}
  The result follows by \ifshort{}{plugging in $a = \y$ and $a = \z$
  and }applying \cref{lem:rank_yz}.
\end{proof}

\subsection{Bounds on the \mgittins{} Rank Function}
\label{sub:rank_bounds_mgittins}

In this section we show two bounds
on $\y[\mgittins{}]$ and $\z[\mgittins{}]$,
each subject to a different assumption on the job size distribution.

\begin{theorem}
  \label{thm:yz_bound_mgittins_or}
  If $X \in \OR{1}{\infty}$, then
  \begin{align*}
    \y[\mgittins{}] &= \Theta(x), \ifshort{&}{\\}
    \z[\mgittins{}] &= \Theta(x).
  \end{align*}
\end{theorem}

\begin{theorem}
  \label{thm:yz_bound_mgittins_qdhr}
  If $X \in \QDHR$ with exponent~$\gamma$, then
  \begin{align*}
    \y[\mgittins{}] &= \Omega(x^{1/\gamma}), \ifshort{&}{\\}
    \z[\mgittins{}] &= O(x^\gamma).
  \end{align*}
\end{theorem}

These bounds are harder to prove than
their \mserpt{} counterparts from \cref{sub:rank_bounds_mserpt}.
The most important component is the following definition,
which helps us better understand the \mgittins{} rank function
and relate it to the simpler \mserpt{} rank function.

\begin{definition}
  \label{def:eta}
  The \emph{time per completion} over\ifshort{}{ an age} interval $(a,b]$ is\footnote{%
    Our time per completion function is the reciprocal of what
    \citet{m/g/1_gittins_aalto, mlps_gittins_aalto} call the
    \emph{efficiency function}.}
  \begin{align*}
    \eta(a, b)
    = \frac{\E{\min\{X, b\} - a \given X > a}}{\P{X < b \given X > a}}
    = \frac{\int_a^b \F{t} \d{t}}{\F{a} - \F{b}}.
  \end{align*}
  We extend this definition to the $b \to a$ and $b \to \infty$ limits:
  \begin{align*}
      \eta(a, a) &= \ifshort{h(a)^{-1}}{\frac{1}{h(a)}}, \ifshort{&}{\\}
    \eta(a, \infty) &= \E{X - a \given X > a}.
  \end{align*}
\end{definition}

We can write the rank functions of
SERPT, \mserpt{}, Gittins, and \mgittins{} in terms of $\eta$ as
\begin{align}[c]
  \label{eq:rank_eta}
  \rank{\serpt{}}{a} &= \eta(a, \infty), \ifshort{&}{\\}
  \rank{\mserpt{}}{a} &= \max_{b \in [0, a]} \eta(b, \infty), \\
  \rank{\gittins{}}{a} &= \min_{b \in [a, \infty]} \eta(a, b), \ifshort{&}{\\}
  \rank{\mgittins{}}{a} &= \max_{b \in [0, a]} \min_{c \in [b, \infty]} \eta(b, c).
\end{align}

Armed with \cref{def:eta, eq:rank_eta}, we are ready to prove
\cref{thm:yz_bound_mgittins_or, thm:yz_bound_mgittins_qdhr}.
The former proof relies on some technical lemmas
that we defer to \cref{sub:rank_bounds_eta}.

\begin{proof}[Proof of \cref{thm:yz_bound_mgittins_or}]
  Throughout this proof, $\y$ and $\z$ refer to $\y[\mgittins{}]$ and~$\z[\mgittins{}]$,
  respectively.
  By \cref{eq:yxz}, it suffices to show
  there exist $C_0, x_0 > 0$ such that \ifshort{$\z \leq C_0 \y$ }{}for all $x \geq x_0$\ifshort{.}{,
  \begin{align*}
    \z \leq C_0 \y.
  \end{align*}}
  We will set $C_0 \geq 2$,
  which covers the $\z \leq 2\y$ case.
  The rest of the proof is thus devoted to the $\z > 2\y$ case.
  Our approach is to show there exist $C_1, C_2$ such that for all $x \geq x_0$,
  \begin{align}
    \label{eq:rank_gittins_goal}
    C_1 \y \geq \rank{\gittins{}}{\y} \geq C_2 \z.
  \end{align}

  We begin with the upper bound on $\rank{\gittins{}}{\y}$.
  By \cref{lem:rank_yz},
  we have $\rank{\gittins{}}{\y} = \rank{\mgittins{}}{\y}$
  for all sizes~$x$,
  and by \cref{eq:rank_eta},
  we have $\rank{\mgittins{}}{a} \leq \rank{\mserpt{}}{a}$
  for all ages~$a$.
  Combining these observations with \cref{thm:yz_bound_mserpt_or}
  implies $\rank{\gittins{}}{\y} = O(\y)$
  and thereby implies the desired upper bound from \cref{eq:rank_gittins_goal}.\footnote{%
    This would be more subtle if $\lim_{x \to \infty} \y$ were finite,
    but \cref{thm:yz_bound_mserpt_or} and a result of
    \citet[Proposition~9]{mlps_gittins_aalto}
    imply $\lim_{x \to \infty} \y = \infty$.}

  We now turn to the lower bound on $\rank{\gittins{}}{\y}$.
  This requires \cref{lem:eta_bound_easy, lem:eta_bound_hard},
  which are facts about $\eta$ that we prove in \cref{sub:rank_bounds_eta}.
  Combining \cref{lem:eta_bound_easy}
  with \cref{eq:rank_eta} and the fact that we are in the $\z > 2\y$ case gives us
  \begin{align*}
      \rank{\gittins{}}{\y} = \eta(\y, \z) \geq \eta\ifshort{(\z/2, \z)}{\gp*{\frac{\z}{2}, \z}}.
  \end{align*}
  By \cref{lem:eta_bound_hard},
  there exist $C_2, x_2$ such that for all~$x$ with $\z/2 > x_2$,
  \begin{align*}
      \eta\ifshort{(\z/2, \z)}{\gp*{\frac{\z}{2}, \z}} \geq C_2 \z,
  \end{align*}
  implying the desired lower bound from \cref{eq:rank_gittins_goal}.
\end{proof}

\begin{proof}[Proof of \cref{thm:yz_bound_mgittins_qdhr}]
  Throughout this proof,
  $\y$ and $\z$ refer to $\y[\mgittins{}]$ and~$\z[\mgittins{}]$, respectively.
  By \cref{eq:yxz}, it suffices to show $\z = O(\y^\gamma)$.
  Because $X \in \QDHR$ with exponent~$\gamma$,
  there exists a strictly increasing function $m : \R_+ \to \R_+$
  such that for all sizes~$x$,
  \begin{align*}
      m(x) \leq \ifshort{h(x)^{-1}}{\frac{1}{h(x)}} \leq m(O(x^\gamma)).
  \end{align*}
  We have $\rank{\gittins{}}{\y} \leq 1/h(\y)$ by~\cref{eq:rank_eta},
  and \cref{lem:rank_yz} implies
  $\rank{\gittins{}}{\z} = \rank{\gittins{}}{\y}$, so
  \ifshort{$\rank{\gittins{}}{\z} \leq m(O(\y^\gamma))$.

  }{%
  \begin{align*}
    \rank{\gittins{}}{\z} \leq m(O(\y^\gamma)).
  \end{align*}}
  It remains only to lower bound $\rank{\gittins{}}{\z}$.
  We do so using the observation that for any age~$a$,
  {\ifshort{}{\useamsalign\allowdisplaybreaks}\begin{align*}
    \rank{\gittins{}}{a}
    &= \min_{b \in [a, \infty]} \eta(a, b) \ifshort{}{\\
    &}= \gp[\Bigg]{
        \max_{b \in [a, \infty]} \frac{\int_a^b \F{t} h(t) \d{t}}{\int_a^b \F{t} \d{t}}
      }^{-1} \\
    &\geq \gp[\big]{\sup_{b > a} h(b)}^{-1} \ifshort{}{\\
      &}= \inf_{b > a} \ifshort{h(b)^{-1}}{\frac{1}{h(b)}} \ifshort{}{\\
    &}\geq m(a),
  \end{align*}}%
  where the first inequality follows from viewing the ratio of integrals
  as a weighted average.
  Plugging in $a = \z$ implies $m(\z) \leq m(O(\y^{\gamma}))$,
  so the result follows because $m$ is strictly increasing.
\end{proof}

\subsection{Time per Completion Lemmas}
\label{sub:rank_bounds_eta}

\begin{lemma}
  \label{lem:eta_bound_easy}
  For all sizes~$x$ and ages~$a$, if $\y < a < \z\esub$, then
  \begin{align*}
    \rank{\gittins{}}{\y} = \eta(\y, \z) \ge \eta(a, \z).
  \end{align*}
\end{lemma}

\begin{proof}
  A property of the Gittins index
  \citep[Lemma~2.2]{book_gittins}
  implies\ifshort{ $\rank{\gittins{}}{\y} = \eta(\y, \z)$.}{}\footnote{%
    The proof given by \citet{book_gittins} is in a discrete setting,
    but essentially the same proof carries over to our continuous setting.%
  }\ifshort{}{\begin{align*}
    \rank{\gittins{}}{\y} = \eta(\y, \z).
  \end{align*}}
  In particular, \ifshort{%
    $\eta(\y, a) \ge \eta(\y, \z)$.
    From this, the desired result follows from the basic properties of the
    $\eta$ function \citep[Eq.~(D.3)]{m-serpt_scully}.%
  }{%
    for any $a \ne \z$,
        \begin{align}
              \label{eq:eta_subopt}
              \eta(\y, a) \ge \eta(\y, \z).
        \end{align}
        A basic property of the $\eta$ function \citep[Eq.~(D.3)]{m-serpt_scully}
        is that for any $d < e < f$,
        \begin{align*}
              \eta(d, e) \ge \eta(d, f) \,\Leftrightarrow\, \eta(d, f) \ge \eta(e, f).
        \end{align*}
        Plugging in $d=\y$, $e=a$, and $f=\z$
        and applying \cref{eq:eta_subopt}
        yields $\eta(\y, \z) \ge \eta(a, \z)$, as desired.%
    }
\end{proof}

\begin{lemma}
  \label{lem:eta_bound_hard}
  If $X \in \OR{1}{\infty}$,
  then there exist constants $C_0, x_0 > 0$
  such that for all $b > a > x_0\esub$,
  \begin{align*}
    \eta(a, b) \ge C_0 a \gp*{1 - \ifshort{\slashfrac}{\frac}{a}{b}}.
  \end{align*}
\end{lemma}

\begin{proof}
  We can write $\eta(a, b)$ as
  \begin{align*}
    \eta(a, b)
    = \frac{
        \int_a^b \F{t}/\F{a} \d{t}
      }{
        1 - \F{b}/\F{a}
      }
    \geq \int_a^b \frac{\F{t}}{\F{a}} \d{t}.
  \end{align*}
  Because $X \in \OR{1}{\infty}$,
  there exist $\beta > 1$ and $C_1, x_1 > 0$
  such that \ifshort{\(}{for all $t > a > x_1$,
  \begin{align*}}
    \ifshort{\slashfrac}{\frac}{\F{t}}{\F{a}}
    \geq C_1\gp*{\ifshort{\slashfrac}{\frac}{t}{a}}^{-\beta}\ifshort{}{.}
  \ifshort{\) for all $t > a > x_1$.}{\end{align*}}
  For all $b > a > x_1$,\ifshort{}{ we have}
  \begin{align*}
    \eta(a, b)
    \geq C_1\int_a^b \gp*{\frac{t}{a}}^{-\beta} \d{t}
    = \frac{C_1 a}{\beta - 1} \gp[\Big]{1-\gp[\Big]{\frac{b}{a}}^{-(\beta-1)}}.
  \end{align*}
  We now consider two cases:
  $\beta \geq 2$ or $1 < \beta < 2$.
  \ifshort{%
    If $\beta \geq 2$, then $(b/a)^{-(\beta - 1)} \leq a/b$.
    If $1 < \beta < 2$,
    then $(b/a)^{-(\beta - 1)} \le 1 + (\beta - 1)(a/b - 1)$.
    In either case, the lemma holds.%
  }{%
    If $\beta \geq 2$, then $(b/a)^{-(\beta - 1)} \leq a/b$ and therefore
    \begin{align}
          \label{eq:eta_or}
          \eta(a, b)
          \geq \frac{C_1a}{\beta-1}\gp[\Big]{1-\frac{a}{b}},
    \end{align}
    so setting $C_0 = C_1/(\beta - 1)$ and $x_0 = x_1$ suffices.
    If $1 < \beta < 2$, we use the fact that for all $u > 0$,
    \begin{align*}
          u^{\beta - 1} \leq 1 + (\beta - 1)(u - 1).
    \end{align*}
    Substituting $u = a/b$ and combining this with \cref{eq:eta_or} yields
    \begin{align*}
          \eta(a,b)
          \geq C_1 a\gp[\Big]{1 - \frac{a}{b}},
    \end{align*}
    so setting $C_0 = C_1$ and $x_0 = x_1$ suffices.%
  }%
\end{proof}

\section{Heavy-Traffic Scaling of \mg{1} Waiting and Residence Times}
\label{sec:heavy}

In this section we characterize the heavy-traffic scaling
of mean waiting, residence, and inflated residence times,
which are the \mg{1} quantities that appear \cref{thm:mgk_response}.
Because \mserpt{} is a simpler policy than \mgittins{},
our approach is to first study \mserpt{}'s heavy-traffic scaling
(\cref{sub:heavy_infinite_variance, sub:heavy_finite_variance})
then show that the results extend to \mgittins{} (\cref{sub:heavy_mgittins}).

\subsection{Key Parts of Waiting and Residence Time}

Before starting the heavy-traffic analyses of \mgittins{} and \mserpt{},
we introduce some new notation.
Let
\begin{align*}
  \H{x} = \ifshort{\slashfrac}{\frac}{\F{x}}{\coload{x}}.
\end{align*}

\begin{definition}
  \useamsalign\allowdisplaybreaks
  The \emph{key \mg{1} response time quantities}, or simply ``key quantities'',
  of a monotonic SOAP policy \generic{} are the following:
  \begin{align*}
    \Qa[\generic{}]
    &= \int_0^\infty
      \gp[\big]{\H{\y[\generic{}]} + \H{\z[\generic{}]}}
      \frac{\lambda \excess{\z[\generic{}]} \F{x}}{\coload{x}^2} \d{x}, \\
    \Qb[\generic{}]
    &= \int_0^\infty \lambda x \H{\y[\generic{}]}^2 \cdot \frac{\F{x}}{\F{\y[\generic{}]}} \d{x}, \\
    \Rb[\generic{}]
      &=\ifshort{ \Sb[\generic{}] =}{}
      \int_0^\infty \lambda \z[\generic{}] \H{\y[\generic{}]}\H{\z[\generic{}]} \cdot \frac{\F{x}}{\F{\y[\generic{}]}} \d{x}, \\
    \Rc[\generic{}]
      &= \int_0^\infty \H{\y[\generic{}]} \cdot \frac{\F{x}}{\F{\y[\generic{}]}} \d{x}, \ifshort{}{\\}
    \ifshort{\qquad}{\Sb[\generic{}]
            &= \Rb[\generic{}], \\}
    \Sc[\generic{}]
    \ifshort{}{&}= \int_0^\infty \H{\y[\generic{}]} \d{x}.
  \end{align*}\ifshort{}{
  When the policy in question is clear,
  we drop the superscript~\generic{}.}
\end{definition}

In \cref{thm:waiting_monotonic, thm:residence_monotonic, thm:sesidence_monotonic}
(\cref{app:new_formulas})
we show that for any monotonic SOAP policy~\generic{},
\ifshort{%
  $\E{\waiting{\generic{}}{}} = \Qa[\generic{}] + \Qb[\generic{}]$,
  $\E{\residence{\generic{}}{}} = \Rb[\generic{}] + \Rc[\generic{}]$,
  and $\E{\sesidence{\generic{}}{}} = \Sb[\generic{}] + \Sc[\generic{}]$.
}{%
\begin{align*}
  \E{\waiting{\generic{}}{}} &= \Qa[\generic{}] + \Qb[\generic{}], \\
  \E{\residence{\generic{}}{}} &= \Rb[\generic{}] + \Rc[\generic{}], \\
  \E{\sesidence{\generic{}}{}} &= \Sb[\generic{}] + \Sc[\generic{}].
\end{align*}
}%
Bounding mean waiting, residence, and inflated residence times
thus amounts to bounding the key quantities.

For the most of the rest of this section
we focus on the case where \generic{} is \mserpt{},
deferring the \mgittins{} case to \cref{sub:heavy_mgittins}.
Until then, $\y$, $\z$, and the key quantities
are understood to have an implicit superscript \mserpt{}.

The most important step of bounding the key quantities is
bounding $\H{\y}$ and $\H{\z}$.
As a first step, we bound $\H{x}$.
Let
\begin{align}
  \label{eq:G}
  \G{x} = \frac{1}{\E{X}} \int_x^\infty \F{t} \d{t}
\end{align}
be the tail of the excess of~$X$.
We can write $\coload{x}$ as
\begin{align}
  \label{eq:coload_G}
  \coload{x} = (1 - \rho) + \rho\G{x}.
\end{align}
This means that for all $\epsilon \in [0, 1]$, we have
\begin{align}
  \label{eq:Hx_bound}
  \H{x}
  \leq \ifshort{}{\frac}{\F{x}}{\max\{1 - \rho, \rho\G{x}\}}
  \leq \ifshort{}{\frac}{\F{x}}{(1 - \rho)^{\ifshort{-}{}\epsilon}(\rho\G{x})^{\ifshort{-\gp}{}{1 - \epsilon}}}
  = \ifshort{}{\frac}{\F{x}^\epsilon\H[1]{x}^{1 - \epsilon}}{(1 - \rho)^{\ifshort{-}{}\epsilon}\rho^{\ifshort{-\gp}{}{1 - \epsilon}}},
\end{align}
where $\H[1]{x} = \F{x}/\G{x} = \lim_{\rho \to 1} \H{x}$.
\ifshort{}{%
This bound is useful because it separates $\H{x}$'s dependence on $x$ and~$\rho$:
the numerator depends only on~$x$,
and the denominator depends only on~$\rho$.}
We will typically choose $\epsilon$ to be either $0$ or arbitrarily small.

Having bounded $\H{x}$ in~\cref{eq:Hx_bound},
we now turn to bounding $\H{\y}$ and $\H{\z}$.
Recalling the definition of $\rank{\serpt{}}{}$ (\cref{def:serpt}),
\begin{align*}
  \H[1]{x} = \ifshort{\slashfrac}{\frac}{\F{x}}{\G{x}} = \ifshort{\slashfrac}{\frac}{\E{X}}{\rank{\serpt{}}{x}},
\end{align*}
so \cref{lem:rank_yz} and the monotonicity of $\rank{\mserpt{}}{}$ imply
\begin{align}
  \label{eq:H1yz_rank}
  \H[1]{\y} = \H[1]{\z} = \ifshort{\slashfrac}{\frac}{\E{X}}{\rank{\mserpt{}}{x}} = O(1).
\end{align}
Combining this with \cref{eq:Hx_bound} yields
bounds on $\H{\y}$ and $\H{\z}$,
though the bounds still have $\F{\y}$ and $\F{\z}$ terms.
To better understand $\H{\y}$ and $\H{\z}$,
we need to use our results from \cref{sec:rank_bounds}
in arguments that depend on what class of distributions contains~$X$.
We do this over the course of \cref{sub:heavy_infinite_variance, sub:heavy_finite_variance}.

\subsection{Infinite-Variance Job Size Distributions}
\label{sub:heavy_infinite_variance}

In this section we study the heavy-traffic scaling of \mserpt{}\ifshort{}{'s
waiting, residence, and inflated residence times}
for infinite-variance job size distributions, specifically those in $\OR{1}{2}$.
\ifshort{M}{With that said, m}any of the intermediate results we prove
will also be useful for the finite-variance $\OR{2}{\infty}$ case
(\cref{sub:heavy_finite_variance}).

\ifshort{}{Suppose that $X \in \OR{1}{\infty}$.
}Combining \cref{thm:yz_bound_mserpt_or, eq:H1yz_rank} gives us
\begin{align}[c]
  \label{eq:H1yz_bounds_or}
  \y, \z &= \Theta(x), \ifshort{&}{\\}
  \H[1]{\y}, \H[1]{\z} &= \Theta\ifshort{(x^{-1})}{\gp[\Big]{\frac{1}{x}}}.
\end{align}
This alone is enough to bound all of the key quantities except~$\Qa$.

\begin{lemma}
  \label{lem:small_terms_or}
  Under \mserpt{}, if $X \in \OR{1}{\infty}$, then
  \begin{align*}
    \Qb, \Rb, \Rc, \Sb, \Sc = O\ifshort{(-\log(1 - \rho))}{\gp*{\log\frac{1}{1 - \rho}}}.
  \end{align*}
\end{lemma}

\begin{proof}
  Our approach is to use the fact that, by \cref{eq:coload_excess},
  \begin{align}
    \label{eq:Hx_integral}
    \int_0^\infty \H{x} \d{x}
    = \int_0^\infty \frac{\F{x}}{\coload{x}} \d{x}
    = \frac{\E{X}}{\rho}\log\frac{1}{1 - \rho}.
  \end{align}
  Because $\Rb = \Sb$ and $\Rc \leq \Sc$,
  it suffices to show that the integrands of $\Qb$, $\Sb$, and $\Sc$
  are all $O(\H{x})$.

  We begin by showing that $\Sc$'s integrand is $O(\H{x})$.
  By \cref{eq:H1yz_bounds_or} and the fact that $X \in \OR{1}{\infty}$, we have
  \ifshort{$\F{\y} = \Theta(\F{x})$, so}{\begin{align*}
    \F{\y} = \F{\Theta(x)} = \Theta(\F{x}),
  \end{align*}
  which yields}
  \begin{align}
    \label{eq:Hy_Hx}
    \H{\y}
    = \ifshort{\slashfrac}{\frac}{\F{\y}}{\coload{\y}}
    \leq \ifshort{\slashfrac}{\frac}{\F{\y}}{\coload{x}}
    = \ifshort{\slashfrac}{\frac}{O(\F{x})}{\coload{x}}
    = O(\H{x}).
  \end{align}
  This implies the desired bound for $\Sc$ and~$\Rc$.

  We show $\Sb$'s integrand is $O(\H{x})$
  by applying \cref{eq:Hx_bound}\ifshort{}{ with $\epsilon = 0$},
  \cref{eq:H1yz_bounds_or}, and~\cref{eq:Hy_Hx}:
  \begin{align*}
    \lambda \z \H{\y}\H{\z}
    \leq \lambda \z \H{\y}\H[1]{\z}
    = O(\H{x}).
  \end{align*}
  This implies the desired bound for $\Sb$ and~$\Rb$.
  Similarly,
  \begin{align*}
    \lambda x \H{\y}^2 \cdot \frac{\F{x}}{\F{\y}}
    \leq \lambda x \H{\y}\H[1]{\y}
    = O(\H{x}),
  \end{align*}
  implying the bound for~$\Qb$.
\end{proof}

It remains only to characterize the heavy-traffic scaling of~$\Qa$.
Treating the $\OR{2}{\infty}$ case requires some additional care,
so we defer it to \cref{sub:heavy_finite_variance},
focusing on the $\OR{1}{2}$ case for now.
The first step is to bound~$\excess{x}$.

\begin{lemma}
  \label{lem:excess_bound_or}
  If $X \in \OR{1}{2}$, then
  \ifshort{\(}{\begin{align*}}
    \excess{x} = \Theta(x^2\F{x})\ifshort{}{.}
  \ifshort{\).}{\end{align*}}
\end{lemma}

\begin{proof}
  By \cref{def:or}, there exists $\beta \in (1, 2)$ such that
  \begin{align*}
    \frac{\excess{x}}{\F{x}}
    = \int_0^x \frac{\lambda t\F{t}}{\F{x}} \d{t}
    \leq O(1) \int_0^x t\gp*{\frac{t}{x}}^{-\beta} \d{t}
    = O(x^2),
  \end{align*}
  and similarly for the lower bound.
\end{proof}

\ifshort{}{We now have bounds on every term in $\Qa$'s integrand,
allowing us to bound $\Qa$ and thereby mean response time.}

\begin{theorem}
  \label{thm:heavy_mserpt_or_iv}
  If $X \in \OR{1}{2}$, then in the $\rho \to 1$ limit,\ifshort{
    \begin{align*}
      \E{\waiting{\mserpt{1}}{}}, \E{\waiting{\mserpt{1}}{}}, \E{\response{\mserpt{1}}{}}
      = O(-\log(1 - \rho)).
    \end{align*}
  }{
  \begin{align*}
    \E{\waiting{\mserpt{1}}{}}
    &= O\gp*{\log\frac{1}{1 - \rho}}, \ifshort{&}{\\}
    \E{\residence{\mserpt{1}}{}}
    &= O\gp*{\log\frac{1}{1 - \rho}},
  \end{align*}
  and therefore
  \begin{align*}
    \E{\response{\mserpt{1}}{}} = O\gp*{\log\frac{1}{1 - \rho}}.
  \end{align*}}
\end{theorem}

\begin{proof}
  \ifshort{%
    By \cref{lem:small_terms_or}, it suffices to show $\Qa = O(-\log(1 - \rho))$.
    After using \cref{lem:excess_bound_or} to bound $\excess{\z}$,
    this follows by a computation similar to that in
    the proof of \cref{lem:small_terms_or}.%
  }{%
    \useamsalign\allowdisplaybreaks
    By \cref{lem:small_terms_or}, it suffices to upper bound~$\Qa$.
    We compute
    \begin{align*}
      \TwoColEqMoveLeft{
        \gp[\big]{\H{\y} + \H{\z}}
        \frac{\lambda \excess{\z} \F{x}}{\coload{x}^2}
      }
      &\leq \gp[\big]{\H[1]{\y} + \H[1]{\z}}
        \frac{\lambda \excess{\z} \H[1]{x}}{\coload{x}}
        \byref{eq:Hx_bound} \\
      &= \gp[\big]{\H[1]{\y} + \H[1]{\z}} \frac{O(\z^2 \F{\z}) \cdot \H[1]{x}}{\coload{x}}
        \byref{lem:excess_bound_or} \\
      &= \frac{O(\F{x})}{\coload{x}}
      \ifshort{= O(\H{x}),}{}
        \byref{eq:H1yz_bounds_or} \ifshort{}{\\*
      &= O(\H{x}),}
    \end{align*}
    so \cref{eq:Hx_integral} implies the desired bound.%
  }
\end{proof}

\subsection{Finite-Variance Job Size Distributions}
\label{sub:heavy_finite_variance}

We now turn to finite-variance job size distributions,
specifically those in $\OR{2}{\infty}$, $\Gumbel$, and $\ENBUE$.
We begin with the simplest case, which is $\ENBUE$.

\begin{theorem}
  \label{thm:heavy_mserpt_enbue}
  If $X \in \ENBUE$, then in the $\rho \to 1$ limit,
  \begin{align*}
    \E{\waiting{\mserpt{1}}{}} &= \Theta\ifshort{((1 - \rho)^{-1})}{\gp*{\frac{1}{1 - \rho}}}, \ifshort{&}{\\}
    \E{\residence{\mserpt{1}}{}} &= \Theta(1),
  \end{align*}
  and therefore
  \ifshort{$\E{\response{\mserpt{1}}{}} = \Theta((1 - \rho)^{-1})$.}{%
  \begin{align*}
    \E{\response{\mserpt{1}}{}} = \Theta\gp*{\frac{1}{1 - \rho}}.
  \end{align*}}
  If additionally $X \in \Bounded$, then\ifshort{
  $\E{\sesidence{\mserpt{1}}{}} = \Theta(1)$}{}
  in the $\rho \to 1$ limit\ifshort{.}{,
  \begin{align*}
    \E{\sesidence{\mserpt{1}}{}} = \Theta(1).
  \end{align*}}
\end{theorem}

\begin{proof}
  Let $x_{\max}$ be the supremum of $X$'s support,
  so we may have $x_{\max} = \infty$.
  Because $X \in \ENBUE$,
  there exists age $a_* < x_{\max}$ such that
  \begin{itemize}
  \item
    $\rank{\mserpt{}}{a} < \rank{\mserpt{}}{a_*}$
    for all $a < a_*$, and
  \item
    $\rank{\mserpt{}}{a} = \rank{\mserpt{}}{a_*}$
    for all $a \geq a_*$.
  \end{itemize}
  This means
  \begin{itemize}
  \item
    $\y \leq a_*$ for all sizes~$x$,
  \item
    $\z \leq a_*$ for all sizes $x \leq a_*$, and
  \item
    $\z = x_{\max}$ for all sizes $x > a_*$.
  \end{itemize}
  Because
  \ifshort{\(}{\begin{align*}}
    \coload{a_*} < \coload{x_{\max}} = 1 - \rho\ifshort{}{,}
  \ifshort{\),}{\end{align*}}
  applying \cref{eq:waiting_residence_monotonic} yields
  \begin{align*}
    \E{\waiting{\mserpt{1}}{}}
    &= \Theta(1)
      + \int_{a_*}^\infty \frac{\excess{x_{\max}}}{\coload{a_*} \cdot (1 - \rho)} \dF{x}
    = \Theta\ifshort{((1 - \rho)^{-1})}{\gp*{\frac{1}{1 - \rho}}},
    \\
    \E{\residence{\mserpt{1}}{}}
    &= \Theta(1)
      + \int_{a_*}^\infty \frac{x}{\coload{a_*}} \dF{x}
    = \Theta(1).
  \end{align*}

  If additionally $X \in \Bounded$, then $x_{\max} < \infty$, so
  \begin{align*}
    \E{\sesidence{\mserpt{1}}{}}
    = \Theta(1) + \int_{a_*}^\infty \frac{x_{\max}}{\coload{a_*}} \dF{x}
    = \Theta(1).
    \qedhere
  \end{align*}
\end{proof}

We now turn to the $\OR{2}{\infty}$ and $\Gumbel$ cases,
which require the following technical lemma.

\begin{lemma}
  \label{lem:H_Ginv_or}
  Let \ifshort{\(}{\begin{align*}}
    L^{\generic{}}(u)
    = \ifshort{1/}{\frac{1}}{\rank{\generic{}}[\big]{\Ginv{1/u}}}\ifshort{}{,}
  \ifshort{\),}{\end{align*}}
  where \generic{} is SERPT or \mserpt{}.
  If $X \in \OR{2}{\infty}$, then
  \begin{align*}
    L^{\serpt{}}, L^{\mserpt{}} \in \OR{0}{1},
  \end{align*}
  and if $X \in \Gumbel$, then
  \begin{align*}
    L^{\serpt{}}, L^{\mserpt{}} \in \ORname(-\epsilon, \epsilon) \text{ for all } \epsilon > 0.
  \end{align*}
\end{lemma}

\begin{proof}
  Because $L^{\mserpt{}}$ is the nonincreasing envelope of $L^{\serpt{}}$,
  it suffices to prove the result for $L^{\serpt{}}$.
  The $\OR{2}{\infty}$ case follows from
  closure properties of Matuszewska indices
  \citep[Lemmas~4.5 and~4.6]{fb_heavy_zwart}.
  The $\Gumbel$ case follows from a result of
  \citet[Section~4.2.2]{fb_heavy_zwart}
  which states that if $X \in \Gumbel$,
  then $L^{\serpt{}}$ is \emph{slowly varying},
  a property implying
  $L^{\serpt{}} \in \ORname(-\epsilon, \epsilon)$ for all $\epsilon > 0$
  \citep{book_bingham}.
\end{proof}

One implication of \cref{lem:H_Ginv_or} is that
if $X \in \Gumbel$, then
\begin{align}
  \label{eq:H_G_gumbel}
  \H[1]{x} = O(\F{x}^{-\epsilon}) \text{ for all } \epsilon > 0.
\end{align}

We are now ready to tackle the $\OR{2}{\infty}$ and $\Gumbel$ cases.
As in \cref{sub:heavy_infinite_variance},
we begin by bounding the five key quantities other than~$\Qa$.
\Cref{lem:small_terms_or} does so for $\OR{2}{\infty}$,
and the following lemma does so for $\Gumbel$.

\begin{lemma}
  \label{lem:small_terms_gumbel}
  Under \mserpt{}, if $X \in \Gumbel$, then
  \begin{align*}
    \Qb, \Rb, \Rc, \Sb = O\ifshort{((1 - \rho)^{-\epsilon})}{\gp*{\frac{1}{(1 - \rho)^\epsilon}}} \text{ for all } \epsilon > 0.
  \end{align*}
  If additionally $X \in \Gumbel \cap (\QDHR \cup \QIMRL)$, then
  \begin{align*}
    \Sc = O\ifshort{((1 - \rho)^{-\epsilon})}{\gp*{\frac{1}{(1 - \rho)^\epsilon}}} \text{ for all } \epsilon > 0.
  \end{align*}
\end{lemma}

\begin{proof}
  \useamsalign\allowdisplaybreaks
  Our overall approach is to use
  \cref{eq:Hx_bound} on each key quantity
  to bound it by an expression of the form
  $(1 - \rho)^{-\epsilon} \cdot \int_0^\infty \Phi(\epsilon, x) \d{x}$,
  where $\Phi(\epsilon, x)$ does not depend on~$\rho$.
  The challenge is then to show that the integral converges
  for arbitrarily small $\epsilon > 0$.

  We begin with two bounds on $\H{\y} \cdot \F{x}/\F{\y}$,
  a term which appears in the integrands of several key quantities.
  By \cref{eq:yxz},
  \begin{align}
    \label{eq:Hy_bound_simple}
    \H{\y} \cdot \ifshort{\slashfrac}{\frac}{\F{x}}{\F{\y}} &\leq \H{\y}, \\
    \label{eq:Hy_bound}
    \H{\y} \cdot \ifshort{\slashfrac}{\frac}{\F{x}}{\F{\y}}
    &= \ifshort{\slashfrac}{\frac}{\F{x}}{\coload{\y}}
    \leq \ifshort{\slashfrac}{\frac}{\F{x}}{\coload{x}}
    = \H{x}.
  \end{align}
  Combining \cref{eq:Hy_bound} with \cref{eq:Hx_integral}
  implies the desired bound for~$\Rc$.

  We now bound~$\Qb$.
  To do so, we apply \cref{eq:Hx_bound} twice,
  choosing $\epsilon = 0$ for $\H{\y}$
  and arbitrarily small $\epsilon > 0$ for $\H{x}$:
  \begin{align*}
    \Qb
    &\leq \int_0^\infty \lambda x \H{\y} \H{x} \d{x}
      \byref{eq:Hy_bound} \\
    &\leq \frac{1}{(1 - \rho)^\epsilon}
      \int_0^\infty \lambda x \F{x}^\epsilon \H[1]{\y} \H[1]{x}^{1 - \epsilon} \d{x}
      \byref{eq:Hx_bound} \\
    &\leq \frac{O(1)}{(1 - \rho)^\epsilon}
      \int_0^\infty x \F{x}^\epsilon \F{x}^{-\epsilon(1 - \epsilon)} \d{x}
      \byref{eq:H1yz_rank, eq:H_G_gumbel} \\
    &\leq \frac{O(1)}{(1 - \rho)^\epsilon}
      \int_0^\infty x^{1 - \alpha\epsilon^2} \d{x},
      \byref{lem:F_gumbel}
  \end{align*}
  where we may choose $\alpha > 0$ arbitrarily large.
  Choosing $\alpha > 2/\epsilon^2$ makes the integral converge,
  so $\Qb = O((1 - \rho)^{-\epsilon})$.
  The computation for~$\Sb$ is similar:
  \begin{align*}
    \Sb
    &\leq \frac{1}{(1 - \rho)^\epsilon}
      \int_0^\infty \lambda \z \F{\z}^\epsilon \H[1]{\y} \H[1]{\z}^{1 - \epsilon} \d{x}
      \byref{eq:Hy_bound_simple, eq:Hx_bound} \\
    &\leq \frac{O(1)}{(1 - \rho)^\epsilon}
      \int_0^\infty \z^{1 - \alpha\epsilon} \d{x}.
      \byref{eq:H1yz_rank, lem:F_gumbel}
  \end{align*}
  Because $\z \geq x$, the integral converges if we choose $\alpha > 2/\epsilon$,
  so $\Sb = O((1 - \rho)^{-\epsilon})$.
  This also covers $\Rb$ because $\Rb = \Sb$.

  If additionally $X \in \Gumbel \cap (\QDHR \cup \QIMRL)$ with exponent~$\gamma$,
  then we can similarly bound~$\Sc$:
  \begin{align*}
    \Sc
    &\leq \frac{1}{(1 - \rho)^\epsilon}
      \int_0^\infty \F{\y}^{\epsilon} \H[1]{\y}^{1 - \epsilon} \d{x}
      \byref{eq:Hx_bound} \\
    &\leq \frac{O(1)}{(1 - \rho)^\epsilon}
      \int_0^\infty \y^{-\alpha\epsilon} \d{x}
      \byref{eq:H1yz_rank, lem:F_gumbel} \\
    &\leq \frac{O(1)}{(1 - \rho)^\epsilon}
      \int_0^\infty x^{-\alpha\epsilon/\gamma} \d{x},
      \byref{thm:yz_bound_mserpt_qimrl}
  \end{align*}
  so choosing $\alpha > \gamma/\epsilon$ shows that
  $\Sc = O((1 - \rho)^{-\epsilon})$.
\end{proof}

\ifshort{}{It remains only to characterize the heavy-traffic scaling of~$\Qa$.}

\begin{lemma}
  \label{lem:big_term_nfv}
  Under \mserpt{}, if $X \in \OR{2}{\infty} \cup \Gumbel$, then
  \begin{align*}
    \Qa = \ifshort{\gp[\Big]}{\gp*}{\ifshort{\gp[\Big]}{\frac{1}}{(1 - \rho) \cdot \rank{\mserpt{}}[\big]{\Ginv{1 - \rho}}}\ifshort{\sp{-1}}{}}.
  \end{align*}
\end{lemma}

\begin{proof}
  \useamsalign\allowdisplaybreaks
  \ifshort{}{Because }$\E{X^2} < \infty$\ifshort{ and thus}{, we have} $\excess{x} = \Theta(1)$,
  so by \cref{eq:Hx_bound, eq:H1yz_rank},
  \begin{equation*}
    \Qa
    = \int_0^\infty
        \frac{\Theta(1)}{\rank{\mserpt{}}{x}}
        \cdot \frac{\F{x}}{\coload{x}^2}
      \d{x}.
  \end{equation*}

  For the lower bound,
  we integrate up to $\Ginv{1 - \rho}$ instead of~$\infty$.
  For $x \leq \Ginv{1 - \rho}$, we have $\G{x} \geq 1 - \rho$,
  so \cref{eq:coload_G} implies
  \begin{equation*}
    \rho \G{x} \leq \coload{x} \leq (1 + \rho) \G{x}.
  \end{equation*}
  Using this fact along with the monotonicity of $\rank{\mserpt{}}{}$
  yields
  \begin{align*}
    \Qa
    &\geq \frac{\Omega(1)}{\rank{\mserpt{}}[\big]{\Ginv{1 - \rho}}}
      \int_0^{\Ginv{1 - \rho}} \frac{\F{x}}{\G{x}^2} \d{x} \\
    &= \frac{\Omega(1)}{\rank{\mserpt{}}[\big]{\Ginv{1 - \rho}}}
      \gp*{\frac{1}{\G[\big]{\Ginv{1 - \rho}}} - 1}
      \byref{eq:G} \\
    &= \Omega\ifshort{\gp[\Big]}{\gp*}{\ifshort{\gp[\Big]}{\frac{1}}{(1 - \rho) \cdot \rank{\mserpt{}}[\big]{\Ginv{1 - \rho}}}\ifshort{\sp{-1}}{}}.
  \end{align*}

  For the upper bound,
  we split the integration region at $\Ginv{1 - \rho}$:
  \begin{align}
    \label{eq:integral_parts}
    \Qa
    &= \int_0^{\Ginv{1 - \rho}}
          \frac{O(1)}{\rank{\mserpt{}}{x}}
          \cdot \frac{\F{x}}{\coload{x}^2}
        \d{x}
      \TwoColEqBreak
      + \int_{\Ginv{1 - \rho}}^\infty
          \frac{O(1)}{\rank{\mserpt{}}{x}}
          \cdot \frac{\F{x}}{\coload{x}^2}
        \d{x}.
  \end{align}
  The second integral in \cref{eq:integral_parts} is simple to bound
  using the monotonicity of $\rank{\mserpt{}}{}$:
  \begin{align*}
    \EqMoveLeft{
      \int_{\Ginv{1 - \rho}}^\infty
          \frac{O(1)}{\rank{\mserpt{}}{x}}
          \cdot \frac{\F{x}}{\coload{x}^2}
        \d{x}
    }
    &\leq \frac{O(1)}{\rank{\mserpt{}}[\big]{\Ginv{1 - \rho}}}
      \int_{\Ginv{1 - \rho}}^\infty \frac{\F{x}}{\coload{x}^2} \d{x} \\
    &\leq \frac{O(1)}{\rank{\mserpt{}}[\big]{\Ginv{1 - \rho}}}
      \gp*{\frac{1}{1 - \rho} - \frac{1}{1 - \rho + \rho\Ginv{1 - \rho}}}
      \TwoColEqBreak[\hspace{20em}]
      \byref{eq:coload_excess, eq:coload_G}
    \ifshort{\\[-\baselineskip]}{\\}
    &= O\ifshort{\gp[\Big]}{\gp*}{\ifshort{\gp[\Big]}{\frac{1}}{(1 - \rho) \cdot \rank{\mserpt{}}[\big]{\Ginv{1 - \rho}}}\ifshort{\sp{-1}}{}}.
  \end{align*}
  To bound the first integral in~\cref{eq:integral_parts},
  we change variables to $u = 1/\G{x}$:
  \begin{align*}
    \TwoColEqMoveLeft{
      \int_0^{\Ginv{1 - \rho}}
        \frac{O(1)}{\rank{\mserpt{}}{x}}
        \cdot \frac{\F{x}}{\coload{x}^2}
      \d{x}
    }
    &\leq \int_0^{\Ginv{1 - \rho}}
        \frac{O(1)}{\rank{\mserpt{}}{x}}
        \cdot \frac{\F{x}}{\G{x}^2}
      \d{x}
      \byref{eq:coload_G} \\
      &= \int_1^{\ifshort{\frac{1}{1 - \rho}}{1/(1 - \rho)}}
        \frac{O(1)}{\rank{\mserpt{}}[\big]{\Ginv{1/u}}}
      \d{u} \ifshort{}{\\ &}
      = O(1)\int_1^{\ifshort{\frac{1}{1 - \rho}}{1/(1 - \rho)}} L^{\mserpt{}}(u) \d{u},
  \end{align*}
  where $L^{\mserpt{}}$ is as in \cref{lem:H_Ginv_or}.
  By \cref{lem:H_Ginv_or},
  we have $L^{\mserpt{}} \in \ORname(-1, \infty)$,
  so a result in Karamata theory \citep[Theorem~2.6.1]{book_bingham} implies
  \begin{equation*}
    \int_1^v L^{\mserpt{}}(u) \d{u} = O(vL^{\mserpt{}}(v))
  \end{equation*}
  in the $v \to \infty$ limit.
  Letting $v = 1/(1 - \rho)$ yields the desired bound.
\end{proof}

\ifshort{}{Having characterized the heavy-traffic scaling of all the key quantities,
the main heavy-traffic results for $\OR{2}{\infty}$ and $\Gumbel$ follow easily.}

\begin{theorem}
  \label{thm:heavy_mserpt_or_fv}
  If $X \in \OR{2}{\infty}$, then in the $\rho \to 1$ limit,
  \begin{align*}
    \E{\waiting{\mserpt{1}}{}}
    &= \Theta\ifshort{\gp[\Big]}{\gp*}{\ifshort{\gp[\Big]}{\frac{1}}{(1 - \rho) \cdot \rank{\mserpt{}}[\big]{\Ginv{1 - \rho}}}\ifshort{\sp{-1}}{}} \\
      &= \Omega\ifshort{((1 - \rho)^{-\delta})}{\gp*{\frac{1}{(1 - \rho)^\delta}}}
      \text{ for some } \delta > 0,
    \\
    \E{\residence{\mserpt{1}}{}}
    & \leq \E{\sesidence{\mserpt{1}}{}} \ifshort{}{\\
    &}= \Theta\gp*{\ifshort{-\log(1 - \rho)}{\log\frac{1}{1 - \rho}}},
  \end{align*}
  and therefore\ifshort{ $\E{\response{\mserpt{1}}{}} = \E{\waiting{\mserpt{1}}{}} + o(\E{\waiting{\mserpt{1}}{}})$.}{
  \begin{align*}
    \E{\response{\mserpt{1}}{}}
    = \Theta\gp*{\frac{1}{(1 - \rho) \cdot \rank{\mserpt{}}[\big]{\Ginv{1 - \rho}}}}.
  \end{align*}}
\end{theorem}

\begin{proof}
  After applying \cref{lem:small_terms_or, lem:big_term_nfv},
  it remains only to show $\Qa = \Omega((1 - \rho)^{-\delta})$.
  Using $L^{\mserpt{}}$ from \cref{lem:H_Ginv_or},
  we can rewrite \cref{lem:big_term_nfv} as
  \begin{align}
    \label{eq:Qa_H_Ginv}
    \Qa
    = \Theta\ifshort{
      \gp[\big]{(1 - \rho)^{-1} L^{\mserpt{}}\gp{(1 - \rho)^{-1}}}
    }{
      \gp[\bigg]{
        \frac{1}{1 - \rho}
        L^{\mserpt{}}\gp[\bigg]{\frac{1}{1 - \rho}}
      }
    }.
  \end{align}
  By \cref{lem:H_Ginv_or}, we have $L \in \OR{0}{1}$,
  which means there exists $\beta \in (0, 1)$
  such that $L(u) = \Omega(u^{-\beta})$ in the $u \to \infty$ limit.
  Letting $\delta = 1 - \beta$
  and $u = 1/(1 - \rho)$ yields the desired bound.
\end{proof}

\begin{theorem}
  \label{thm:heavy_mserpt_gumbel}
  If $X \in \Gumbel$, then in the $\rho \to 1$ limit,
  \begin{align*}
    \E{\waiting{\mserpt{1}}{}}
    &= \Theta\ifshort{\gp[\Big]}{\gp*}{\ifshort{\gp[\Big]}{\frac{1}}{(1 - \rho) \cdot \rank{\mserpt{}}[\big]{\Ginv{1 - \rho}}}\ifshort{\sp{-1}}{}} \\
    &= \Omega\ifshort{((1 - \rho)^{-(1 - \epsilon)})}{\gp*{\frac{1}{(1 - \rho)^{1 - \epsilon}}}}
      \text{ for all } \epsilon > 0,
    \\
    \E{\residence{\mserpt{1}}{}}
      &= O\ifshort{((1 - \rho)^{-\epsilon})}{\gp*{\frac{1}{(1 - \rho)^\epsilon}}}
      \text{ for all } \epsilon > 0,
  \end{align*}
  and therefore \ifshort{ $\E{\response{\mserpt{1}}{}} = \E{\waiting{\mserpt{1}}{}} + o(\E{\waiting{\mserpt{1}}{}})$.}{
  \begin{align*}
    \E{\response{\mserpt{1}}{}}
    = \Theta\gp*{\frac{1}{(1 - \rho) \cdot \rank{\mserpt{}}[\big]{\Ginv{1 - \rho}}}}.
  \end{align*}}
  If additionally $X \in \Gumbel \cap (\QDHR \cup \QIMRL)$, then
  \begin{align*}
    \E{\sesidence{\mserpt{1}}{}}
      = O\ifshort{((1 - \rho)^{-\epsilon})}{\gp*{\frac{1}{(1 - \rho)^\epsilon}}}
      \text{ for all } \epsilon > 0.
  \end{align*}
\end{theorem}

\begin{proof}
  After applying \cref{lem:small_terms_gumbel, lem:big_term_nfv},
  it remains only to show $\Qa = \Omega((1 - \rho)^{-(1 - \epsilon)})$.
  This follows from \cref{eq:Qa_H_Ginv, lem:H_Ginv_or},
  similarly to the proof of \cref{thm:heavy_mserpt_or_fv}.
\end{proof}

\subsection{Extending Heavy-Traffic Analysis from \mserpt{}
  to \gittins{} and \mgittins{}}
\label{sub:heavy_mgittins}

Having characterized heavy-traffic scaling under \mserpt{},
we now do the same for \gittins{} and \mgittins{}.
Our first result shows that the mean waiting and residence times
of \gittins{} and \mgittins{} have
the same heavy-traffic scaling as that of \mserpt{}.
Note that the precondition holds for all of the job size distributions
we consider in \cref{sub:heavy_finite_variance}.\footnote{%
  With some extra effort, one can show it also holds for $X \in \OR{1}{2}$.}

\begin{theorem}
  \label{thm:heavy_mgittins_response}
  In the $\rho \to 1$ limit,
  \begin{align*}
    \E{\residence{\gittins{1}}{}}, \E{\residence{\mgittins{1}}{}}
    = O(\E{\residence{\mserpt{1}}{}}),
  \end{align*}
  and if $\E{\residence{\mserpt{1}}{}} = O(\E{\waiting{\mserpt{1}}{}})$, then
  \begin{align*}
    \E{\waiting{\gittins{1}}{}}, \E{\waiting{\mgittins{1}}{}}
    = \Theta(\E{\waiting{\mserpt{1}}{}}).
  \end{align*}
\end{theorem}

\begin{proof}
  The residence time result follows immediately from
  results of \citet[Eq.~(3.8) and Proposition~4.8]{m-serpt_scully},
  which imply
  \begin{align*}
    \E{\residence{\gittins{1}}{}}
    \leq \E{\residence{\mgittins{1}}{}}
    \leq \E{\residence{\mserpt{1}}{}}.
  \end{align*}
  For waiting time, we first invoke further results of
  \citet[Proposition~4.7 and Lemma~5.6]{m-serpt_scully},
  which imply
  \begin{align*}
    \E{\waiting{\gittins{1}}{}}
    \geq \E{\waiting{\mgittins{1}}{}}
    \geq \ifshort{\slashfrac}{\frac}{\E{\waiting{\mserpt{1}}{}}}{2}.
  \end{align*}
  It thus suffices to show
  $\E{\waiting{\gittins{1}}{}} = O(\E{\waiting{\mserpt{1}}{}})$.
  Because Gittins minimizes mean response time
  \citep{m/g/1_gittins_aalto, mlps_gittins_aalto, book_gittins},
  we have
  \begin{align*}
    \E{\waiting{\gittins{1}}{}}
    \leq \E{\response{\gittins{1}}{}}
    \leq \E{\response{\mserpt{1}}{}}
    = \E{\waiting{\mserpt{1}}{}} + \E{\residence{\mserpt{1}}{}},
  \end{align*}
  so the result follows from
  the $\E{\residence{\mserpt{1}}{}} = O(\E{\waiting{\mserpt{1}}{}})$ precondition.
\end{proof}

Our final heavy-traffic result shows that
for certain job size distributions, under \mgittins{},
mean waiting time dominates mean inflated residence time.
The conditions are the same as those shown for \mserpt{}
over the course of \cref{sub:heavy_finite_variance},
except $\QDHR \cup \QIMRL$ is replaced by $\QDHR$.

\begin{theorem}
  \label{thm:heavy_mgittins_sesidence}
  If
  \ifshort{$X$ is in $\OR{2}{\infty}$, $\Gumbel \cap \QDHR$, or $\Bounded$,}{\begin{align*}
    X \in \OR{2}{\infty} \cup (\Gumbel \cap \QDHR) \cup \Bounded,
  \end{align*}}
  then \ifshort{}{in the $\rho \to 1$ limit,
  \begin{align*}
    \E{\sesidence{\mgittins{1}}{}} = o(\E{\waiting{\mgittins{1}}{}}).
  \end{align*}
  More specifically, }$\E{\sesidence{\mgittins{1}}{}}$ obeys
  the same scaling bounds as shown for $\E{\sesidence{\mserpt{1}}{}}$
  in \cref{thm:heavy_mserpt_or_fv, thm:heavy_mserpt_gumbel, thm:heavy_mserpt_enbue}.
\end{theorem}

\begin{proof}
  The proof is very similar to the proofs of analogous results for \mserpt{}
  (\cref{thm:heavy_mserpt_or_fv, thm:heavy_mserpt_gumbel, thm:heavy_mserpt_enbue}),
  so we just describe the differences.
  \begin{itemize}
  \item
    If $X \in \OR{2}{\infty}$,
    we follow the same proof as \cref{thm:heavy_mserpt_or_fv}
    and the lemmas it requires,
    except we use \cref{thm:yz_bound_mgittins_or} to bound
    $\y[\mgittins{}]$ and $\z[\mgittins{}]$.
  \item
    If $X \in \Gumbel \cap \QDHR$,
    we follow the same proof as \cref{thm:heavy_mserpt_gumbel}
    and the lemmas it requires,
    except we use \cref{thm:yz_bound_mgittins_qdhr} to bound
    $\y[\mgittins{}]$ and $\z[\mgittins{}]$.
  \item
    If $X \in \Bounded$,
    we follow the same proof as \cref{thm:heavy_mserpt_enbue},
    except we use a result of \citet[Proposition~9]{mlps_gittins_aalto}
    to justify the existence of the critical age~$a_*$.
    \qedhere
  \end{itemize}
\end{proof}

\section{Conclusion}
\label{sec:conclusion}

We study optimal scheduling in the \mg{k} to minimize mean response time.
This problem is solved by the Gittins policy for the single-server $k = 1$ case
but was previously open for the much more difficult multiserver case.
We introduce a new variant of Gittins called \emph{\mgittins{}}
(\cref{def:mgittins})
and show that it minimizes mean response time in the heavy-traffic \mg{k}
for a large class of \ifshort{}{finite-variance }job size distributions
(\cref{thm:mgittins_opt}).
We also show that the simple and practical \mserpt{} policy
is a $2$\=/approximation for mean response time in the heavy-traffic \mg{k}
under similar conditions
(\cref{thm:mserpt_approx}).
As a byproduct of our \mg{k} study,
we obtain results characterizing the heavy-traffic scaling of \mg{1} mean response time
under Gittins, \mgittins{}, and \mserpt{}
(\cref{thm:heavy}).

A natural question to ask is whether the conditions under which we prove
\mgittins{}'s optimality can be relaxed,
particularly the $\QDHR$ and $\Bounded$ assumptions.
The difficulty lies in the fact that
for some job size distributions,
the bound in \cref{thm:mgk_response} is not strong enough
because inflated residence time is infinite.
It is possible that the techniques used by
\citet{fcfs_heavy_dist_kollerstrom, fcfs_heavy_mean_kollerstrom}
to analyze the heavy-traffic \mg{k} under FCFS could be helpful,
seeing as FCFS has infinite inflated residence time.

Another major open question is analyzing the performance of \mgittins{}
outside of the heavy-traffic limit.
In the single-server case,
one can generalize the techniques of \citet{m-serpt_scully}
to show that \mgittins{} is a $3$\=/approximation for \mg{1} mean response time
at all loads.
However, the multiserver case remains open.

\section*{Acknowledgments}

This work was supported by NSF grants
CMMI\=/1938909, XPS\=/1629444, and CSR\=/1763701;
and a Google 2020 Faculty Research Award.

\ifacmart{}{\section*{References}}

\bibliographystyle{ACM-Reference-Format}
\bibliography{refs}

\appendix
\renewcommand{\thesection}{\Alph{section}}

\ifshort{}{%
\section{Difficulty of \mg{k} Analysis for Nonmonotonic Rank Functions}
\label{app:mgk_nonmonotonic}

In this appendix we explain why \cref{thm:mgk_response}
does not readily generalize to SOAP policies with nonmonotonic rank functions.

Recall that the proof of \cref{thm:mgk_response} considers
a tagged job~$J$ of size~$x$
and considers several categories of work
completed while $J$ is in the system.
Our focus here is on relevant work,
which is work on jobs that are prioritized ahead of~$J$.
Let $\arel$ be the maximum age at which a new job,
namely one that arrives after~$J$,
can contribute relevant work under \generic{k}.
When \generic{} is monotonic, $\arel$ does not depend on the number of servers~$k$.
Specifically, we have $\arel = \y[\generic{}]$.
The problem for nonmonotonic SOAP policies~$\generic{}$ is that, as we show below,
we can have $\arel > \arel[1]$ when $k \geq 2$.

The following discussion
uses definitions of $\y[\generic{}]$ and $\z[\generic{}]$
generalized to all SOAP policies~\generic{}.
\begin{itemize}
\item
  If \generic{} is monotonic,
  then $\y[\generic{}]$ and $\z[\generic{}]$ are given by \cref{def:yz}.
\item
  If \generic{} is nonmonotonic, we can define $\y[\generic{}]$ and $\z[\generic{}]$
  in terms of a monotonic SOAP policy related to \generic{} \citep{m-serpt_scully}.
  Specifically, letting \mgeneric{} be the monotonic SOAP policy with rank function
  \begin{align*}
    \rank{\mgeneric{}}{a} = \max_{b \in [0, a]} \rank{\generic{}}{b},
  \end{align*}
  we define $\y[\generic{}] = \y[\mgeneric{}]$ and $\z[\generic{}] = \z[\mgeneric{}]$.
\end{itemize}

\begin{figure}
  \centering
  \input{fig_yz_nonmonotonic}
  \caption{Age Cutoffs for Nonmonotonic Rank Functions}
  \label{fig:yz_nonmonotonic}
\end{figure}
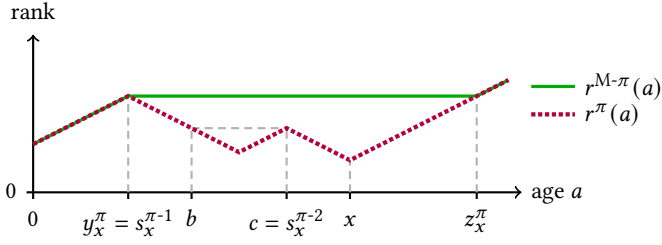

Consider the example SOAP policy~\generic{} and tagged job size~$x$
shown in \cref{fig:yz_nonmonotonic}.
In the single-server $k = 1$ case, we have $\arel[1] = \y[\generic{}]$.
To see why, consider the moment a new job $J'$ reaches age~$\y[\generic{}]$
while the tagged job~$J$ is still in the system.
For this to occur, it must be that $J$ is also at age~$\y[\generic{}]$,
because otherwise $J$ would have priority over~$J'$.
With both $J$ and $J'$ at the same rank, the FCFS tiebreaker prioritizes~$J$.
Thereafter, $J$ never has rank worse than $\rank{\generic{}}{\y[\generic{}]}$,
so $J'$ remains stuck at age~$\y[\generic{}]$
and is never prioritized over~$J$.

We now reconsider the same example from \cref{fig:yz_nonmonotonic}
but with $k \geq 2$ servers.
The key difference is that because there are multiple servers,
$J'$ can receive service even while $J$ has better rank
because $J$ and $J'$ can occupy different servers simultaneously.
This means $J'$ no longer gets stuck at age~$\y[\generic{}]$.
In particular, if $J$ reaches age~$c$ and $J'$ passes age~$b$,
then $J'$ contributes relevant work between ages $b$ and~$c$.
Therefore, $\arel = c > \arel[1]$ for $k \geq 2$.

The bound in \cref{thm:mgk_response} follows from assuming that
every new job~$J'$ will contribute relevant work
until it completes or reaches age~$\arel$.
This is a worst-case estimate,
because the tagged job~$J$ might complete
before $J'$ completes or reaches age~$\arel$.
When \generic{} is monotonic, we have $\arel = \arel[1]$,
so this overestimate is tight enough
to compare the mean response times under \generic{k} and \generic{1}.
However, when \generic{} is nonmonotonic,
it may be that $\arel > \arel[1]$, as explained above,
so we do not obtain a tight comparison between the \generic{k} and \generic{1} systems.
This suggests generalizing \cref{thm:mgk_response} to nonmonotonic SOAP policies
requires not relying as heavily on worst-case quantities like~$\arel$.
}

\section{New Formulas for Mean Waiting and Residence Times}
\label{app:new_formulas}

In this appendix we prove the following new formulas
for mean waiting, residence, and inflated residence times.

\begin{theorem}
  \label{thm:waiting_monotonic}
  Under any monotonic SOAP policy~\generic{},
  \begin{align*}
    \E{\waiting{\generic{1}}{}}
    = \int_0^\infty \gp[\Bigg]{
        \gp[\bigg]{
          \frac{\F{\y[\generic{}]}}{\coload{\y[\generic{}]}} + \frac{\F{\z[\generic{}]}}{\coload{\z[\generic{}]}}
        } \frac{\lambda \excess{\z[\generic{}]} \F{x}}{\coload{x}^2}
        + \frac{\lambda x \F{\y[\generic{}]} \F{x}}{\coload{\y[\generic{}]}^2}
      } \d{x}.
  \end{align*}
\end{theorem}

\begin{theorem}
  \label{thm:residence_monotonic}
  Under any monotonic SOAP policy~\generic{},
  \begin{align*}
    \E{\residence{\generic{1}}{}}
    = \int_0^\infty \gp[\bigg]{
        \frac{\lambda \z[\generic{}] \F{x} \F{\z[\generic{}]}}{\coload{\y[\generic{}]} \coload{\z[\generic{}]}}
        + \frac{\F{x}}{\coload{\y[\generic{}]}}
      } \d{x}.
  \end{align*}
\end{theorem}

\begin{theorem}
  \label{thm:sesidence_monotonic}
  Under any monotonic SOAP policy~\generic{},
  \begin{align*}
    \E{\sesidence{\generic{1}}{}}
    = \int_0^\infty \gp[\bigg]{
        \frac{\lambda \z[\generic{}] \F{x} \F{\z[\generic{}]}}{\coload{\y[\generic{}]} \coload{\z[\generic{}]}}
        + \frac{\F{\y[\generic{}]}}{\coload{\y[\generic{}]}}
      } \d{x}.
  \end{align*}
\end{theorem}

Proving these results requires new technical machinery for, roughly speaking,
performing integration by parts on expressions involving $\y[\generic{}]$ and~$\z[\generic{}]$,
such as those in~\cref{eq:waiting_residence_monotonic}.
\Cref{sub:hills_valleys} introduces the general technical machinery,
which \cref{sub:new_formulas_proofs} then applies to prove the above results.

Throughout this appendix,
$\diff$ denotes the derivative operator,
and $[t_1, \dots, t_n \mapsto \text{RHS}]$ denotes the function
that maps variables $t_1, \dots, t_n$ to expression~$\text{RHS}$.

\subsection{Integration by Parts with Hills and Valleys}
\label{sub:hills_valleys}

\begin{definition}
  A \emph{hill-valley partition} of $\R_+$ is a sequence
  \begin{align*}
    0 = u_0 \leq v_0 < u_1 < v_1 < u_2 < v_2 < \dots.
  \end{align*}
  Intervals of the form $(u_i, v_i]$ are called \emph{valleys},
  and intervals of the form $(v_i, u_{i + 1}]$
  are called \emph{hills}.\footnote{%
    We borrow the terms ``hill'' and ``valley'' from \citet{m-serpt_scully},
    who use a similar concept to analyze SOAP policies,
    but this definition is abstracted away from the details of SOAP.
    As a corner case,
    we consider the first hill or valley to also include~$0$.}
\end{definition}

\begin{definition}
  \label{def:hill-valley_pair}
  Functions $\ySize{}, \zSize{} : \R_+ \to \R_+$ are a \emph{hill-valley pair}
  for a given hill-valley partition if for each valley $(u_i, v_i]$,
  \begin{align*}
    \ySize{}(x) &= u_i, & \zSize{}(x) &= v_i, & \text{for all } x \in (u_i, v_i],
  \end{align*}
  and for each hill $(v_i, u_{i + 1}]$,
  \begin{align*}
    \ySize{}(x) &= x, & \zSize{}(x) &= x, & \text{for all } x \in (v_i, u_{i + 1}].
  \end{align*}
  For compactness, we write $\y = \ySize{}(x)$ and $\z = \zSize{}(x)$.
\end{definition}

It is simple to check that for any monotonic SOAP policy~\generic{},
the pair $\ySize[\generic{}]{}, \zSize[\generic{}]{}$ (\cref{def:yz})
is a hill-valley pair.

\begin{definition}
  \label{def:Diff}
  For functions $\Phi : \R_+ \to \R_+$,
  we define the \emph{difference ratio operator}~$\Diff$ as follows:
  \begin{align*}
    \Diff\Phi(\angle{u, v}) =
    \begin{dcases}
      \frac{\Phi(v) - \Phi(u)}{v - u} & \text{if } u \neq v \\
      \diff\Phi(u) & \text{if } u = v,
    \end{dcases}
  \end{align*}
  where $\diff$ is the derivative operator.
  Similarly, for functions with multiple arguments,
  $\Diff_i$ is a version of $\Diff$ that works on the $i$th argument:
  \begin{align*}
    \Diff_i\Phi(\dots, \angle{u, v}, \dots) =
    \Diff[t \mapsto \Phi(\dots, t, \dots)](\angle{u, v}).
  \end{align*}
\end{definition}

Like $\diff$, it is easily seen that $\Diff$ is a linear operator.
When applied to polynomials, $\Diff$ elegantly generalizes~$\diff$.
For example,
\begin{align}
  \label{eq:Diff_recip}
  \Diff\sqgp[\bigg]{t \mapsto \frac{1}{t}}(\angle{u, v}) = \frac{1}{uv}.
\end{align}
The $\Diff$ operator also obeys various chain-rule-like identities.
We highlight the two we use below.

\begin{lemma}
  \label{lem:Diff_chain}
  Let $\Phi, \Psi : \R \to \R$ be differentiable.
  For all $u, v \in \R$,
  \begin{align*}
    \Diff[t \mapsto \Phi(\Psi(t))](\angle{u, v})
    = \Diff\Phi(\angle{\Psi(u), \Psi(v)}) \Diff\Psi(\angle{u, v}).
  \end{align*}
\end{lemma}

\begin{proof}
  If $u = v$, this is the chain rule.
  If $u \neq v$ but $\Psi(u) = \Psi(v)$, then both sides are~$0$.
  If $\Psi(u) \neq \Psi(v)$, then the result follows by a simple computation.
\end{proof}

\begin{lemma}
  \label{lem:Diff_split}
  Let $\Phi : \R^2 \to \R$ be differentiable.
  For all $u, v \in \R$,
  \begin{align*}
      \Diff[t \mapsto \Phi(t, t)](\angle{u, v})
    = \Diff_2\Phi(u, \angle{u, v}) + \Diff_1\Phi(\angle{u, v}, v).
  \end{align*}
\end{lemma}

\begin{proof}
  If $u = v$, this is the multivariable chain rule.
  If $u \neq v$,
  \begin{align*}
    \TwoColEqMoveLeft{
      (v - u) \Diff[t \mapsto \Phi(t, t)](\angle{u, v})
    }
    &= \Phi(v, v) - \Phi(u, u) \\
    &= \Phi(v, v) - \Phi(u, v) + \Phi(u, v) - \Phi(u, u) \\
    &= (v - u)\gp{
        \Diff_1\Phi(\angle{u, v}, v)
        + \Diff_2\Phi(u, \angle{u, v})
      }.
    \qedhere
  \end{align*}
\end{proof}

The most important result of this appendix is the following lemma,
which formulates a version of integration by parts that
works for hill-valley pairs despite their discontinuity.

\begin{lemma}
  \label{lem:yz_integrate_by_parts}
  Let
  $\ySize{}, \zSize{}$ be a hill-valley pair,
  $\Phi : \R_+^3 \to \R$ be differentiable,
  $P : \R_+ \to \R$ be differentiable, and
  $\ol{P}(x) = c - P(x)$ for some $c \in \R$.
  If
  \begin{align*}
    \ol{P}(0) \Phi(0, 0, \zSize{0}) &= 0, \ifshort{&}{\\}
    \lim_{x \to \infty} \ol{P}(x) \Phi(\y, x, \z) &= 0,
  \end{align*}
  then
  \begin{align*}
    \ifshort{\EqMoveLeft}{\OneColEqMoveLeft}{
      \int_0^\infty \Phi(\y, x, \z) \diff P(x) \d{x}
    }
    &\begin{aligned}[t]
    {}=
      \int_0^\infty \Bigl(
        & \ol{P}(\y) \Diff_3\Phi(\y, \y, \angle{\y, \z})
        \ifshort{}{\TwoColEqBreak[{}]}
        + \ol{P}(x) \diff_2\Phi(\y, x, \z)
        \TwoColEqBreak[{}]
        + \ol{P}(v) \Diff_1\Phi(\angle{\y, \z}, \z, \z)
      \Bigr) \d{x}.
      \end{aligned}
  \end{align*}
\end{lemma}

\begin{proof}
  \useamsalign\allowdisplaybreaks
  For each valley $(u, v]$,
  \begin{align*}
    \EqMoveLeft{
      \int_u^v \Phi(\y, x, \z) \diff P(x) \d{x}
    }
    &= \int_u^v \ol{P}(x) \diff_2\Phi(u, x, v) \d{x}
      + \ol{P}(u)\Phi(u, u, v) - \ol{P}(v)\Phi(u, v, v) \\
    &= \int_u^v \ol{P}(x) \diff_2\Phi(u, x, v) \d{x}
      + \ol{P}(u)\Phi(u, u, u) - \ol{P}(v)\Phi(v, v, v)
    \EqBreak
      + (v - u)\ol{P}(u)\Diff_3\Phi(u, u, \angle{u, v})
      + (v - u)\ol{P}(v)\Diff_1\Phi(\angle{u, v}, v, v) \\
    &\begin{aligned}
    {}=
      \int_u^v \Bigl(
        & \ol{P}(u) \Diff_3\Phi(u, u, \angle{u, v})
        \ifshort{}{\TwoColEqBreak[{}]}
        + \ol{P}(x) \diff_2\Phi(u, x, v)
        \TwoColEqBreak[\ifshort{\hspace{-2.5em}}{}]
        + \ol{P}(v) \Diff_1\Phi(\angle{u, v}, v, v)
      \Bigr) \d{x}
      \ifshort{+ \ol{P}(u)\Phi(u, u, u) - \ol{P}(v)\Phi(v, v, v)}{}
    \end{aligned}
    \ifshort{}{\EqBreak
      + \ol{P}(u)\Phi(u, u, u) - \ol{P}(v)\Phi(v, v, v)} \\
    &\begin{aligned}
    {}=
      \int_u^v \Bigl(
        & \ol{P}(\y) \Diff_3\Phi(\y, \y, \angle{\y, \z})
        \ifshort{}{\TwoColEqBreak[{}]}
        + \ol{P}(x) \diff_2\Phi(\y, x, \z)
        \TwoColEqBreak[\ifshort{\hspace{-2.5em}}{}]
        + \ol{P}(v) \Diff_1\Phi(\angle{\y, \z}, \z, \z)
      \Bigr) \d{x}
      \ifshort{+ \ol{P}(u)\Phi(u, u, u) - \ol{P}(v)\Phi(v, v, v).}{}
    \end{aligned}
    \ifshort{}{\EqBreak
      + \ol{P}(u)\Phi(u, u, u) - \ol{P}(v)\Phi(v, v, v).}
  \end{align*}
  For each hill $(v, u]$,
  \begin{align*}
    \ifshort{\EqMoveLeft}{\OneColEqMoveLeft}{
      \int_v^u \Phi(\y, x, \z) \diff P(x) \d{x}
    }
    &= \int_v^u \ol{P}(x) \diff[t \to \Phi(t, t, t)](x) \d{x}
    \ifshort{}{\TwoColEqBreak}
      + \ol{P}(v)\Phi(v, v, v) - \ol{P}(u)\Phi(u, u, u) \\
    &\begin{aligned}
    {}=
      \int_v^u \Bigl(
        & \ol{P}(x) \diff_3\Phi(x, x, x)
        \ifshort{}{\TwoColEqBreak[{}]}
        + \ol{P}(x) \diff_2\Phi(x, x, x)
        \TwoColEqBreak[\ifshort{\hspace{-2.5em}}{}]
        + \ol{P}(x) \diff_1\Phi(x, x, x)
      \Bigr) \d{x}
      \ifshort{+ \ol{P}(v)\Phi(v, v, v) - \ol{P}(u)\Phi(u, u, u)}{}
    \end{aligned}
    \ifshort{}{\EqBreak
      + \ol{P}(v)\Phi(v, v, v) - \ol{P}(u)\Phi(u, u, u)} \\
    &\begin{aligned}
    {}=
      \int_v^u \Bigl(
        & \ol{P}(\y) \Diff_3\Phi(\y, \y, \angle{\y, \z})
        \ifshort{}{\TwoColEqBreak[{}]}
        + \ol{P}(x) \diff_2\Phi(\y, x, \z)
        \TwoColEqBreak[\ifshort{\hspace{-2.5em}}{}]
        + \ol{P}(v) \Diff_1\Phi(\angle{\y, \z}, \z, \z)
      \Bigr) \d{x}
      \ifshort{+ \ol{P}(v)\Phi(v, v, v) - \ol{P}(u)\Phi(u, u, u).}{}
    \end{aligned}
    \ifshort{}{\EqBreak
      + \ol{P}(v)\Phi(v, v, v) - \ol{P}(u)\Phi(u, u, u).}
  \end{align*}
  Summing the hill and valley expressions over all hills and valleys,
  most of the non-integral terms cancel out,
  and the two that remain are $0$ by assumption\ifshort{.}{:
  \usepreludealign
  \begin{align*}
    \EqMoveLeft{
      \int_0^\infty \Phi(\y, x, \z) \diff P(x) \d{x}
    }
    &\begin{aligned}[t]
    {}=
      \int_0^\infty \Bigl(
        & \ol{P}(\y) \Diff_3\Phi(\y, \y, \angle{\y, \z})
        \TwoColEqBreak[{}]
        + \ol{P}(x) \diff_2\Phi(\y, x, \z)
        \TwoColEqBreak[{}]
        + \ol{P}(v) \Diff_1\Phi(\angle{\y, \z}, \z, \z)
      \Bigr) \d{x}
    \end{aligned}
    \mkern-12mu
    \EqBreak
      + \ol{P}(0)\Phi(0, 0, \zSize{0})
      - \lim_{x \to \infty} \ol{P}(x)\Phi(\y, x, \z).
    \qedhere
  \end{align*}}
\end{proof}

Our final two lemmas show that
integrals using $\Diff$ can sometimes be turned into
integrals using~$\diff$.

\begin{lemma}
  \label{lem:yz_Diff_diff}
  Let $\ySize{}, \zSize{}$ be a hill-valley pair
  and $\Phi : \R_+^3 \to \R_+\esub$ be differentiable
  with respect to its second argument.
  Then
  \begin{align*}
    \int_0^\infty
      \Diff_2\Phi(\y, \angle{\y, \z}, \z)
    \d{x}
    &= \int_0^\infty \diff_2\Phi(\y, x, \z) \d{x}.
  \end{align*}
\end{lemma}

\begin{proof}
  \useamsalign\allowdisplaybreaks
  For each valley $(u, v]$,
  \begin{align*}
    \ifshort{\MoveEqLeft}{}
    \int_u^v \Diff_2\Phi(\y, \angle{\y, \z}, \z) \d{x}
    \ifshort{}{&}= \int_u^v \Diff_2\Phi(u, \angle{u, v}, v) \d{x} \\
    &= (v - u)\Diff_2\Phi(u, \angle{u, v}, v) \ifshort{}{\\
    &}= \Phi(u, v, v) - \Phi(u, u, v) \\
    &= \int_u^v \diff_2\Phi(u, x, v) \d{x} \ifshort{}{\\
    &}= \int_u^v \diff_2\Phi(\y, x, \z) \d{x}.
  \end{align*}
  For each hill $(v, u]$,
  \begin{align*}
    \ifshort{\MoveEqLeft}{}
    \int_v^u \Diff_2\Phi(\y, \angle{\y, \z}, \z) \d{x}
    \ifshort{}{&}= \int_v^u \Diff_2\Phi(x, \angle{x, x}, x) \d{x} \\
    &= \int_v^u \diff_2\Phi(x, x, x) \d{x} \ifshort{}{\\
    &}= \int_v^u \diff_2\Phi(\y, x, \z) \d{x}.
  \end{align*}
  Summing the hill and valley expressions over all hills and valleys
  yields the desired result.
\end{proof}

\begin{lemma}
  \label{lem:yz_Diff_diff_chain}
  Let
  $\ySize{}, \zSize{}$ be a hill-valley pair
  and both $\Phi : \R_+^3 \to \R$ and $\Psi : \R_+ \to \R$ be differentiable.
  Then
  \begin{align*}
    \TwoColEqMoveLeft{
      \int_0^\infty
        \Diff[t \mapsto \Phi(\y, \Psi(t), \z)](\angle{\y, \z})
      \d{x}
    }
    &= \int_0^\infty
        \Diff_2\Phi(\y, \angle{\Psi(\y), \Psi(\z)}, \z) \diff\Psi(x)
      \d{x}.
  \end{align*}
\end{lemma}

\begin{proof}
  \useamsalign\allowdisplaybreaks
  We compute
  \begin{align*}
    \EqMoveLeft{
      \int_0^\infty
        \Diff[t \mapsto \Phi(\y, \Psi(t), \z)](\angle{\y, \z})
      \d{x}
    }
    &= \int_0^\infty
        \Diff_2\Phi(\y, \angle{\Psi(\y), \Psi(\z)}, \z)
        \Diff \Psi(\angle{\y, \z})
      \d{x}
    \byref{lem:Diff_chain}
    \\
    &\begin{aligned}
      {}= \int_0^\infty
        \Diff_2\Bigl[&u, t, v \mapsto
        \ifshort{}{\TwoColEqBreak}
          \Diff_2\Phi(u, \angle{\Psi(u), \Psi(v)}, v)
          \cdot \Psi(t)
        \Bigr](\y, \angle{\y, \z}, \z)
      \d{x}
    \end{aligned}
    \\
    &= \int_0^\infty
        \Diff_2\Phi(\y, \angle{\Psi(\y), \Psi(\z)}, \z) \diff\Psi(x)
      \d{x}.
    \byref{lem:yz_Diff_diff}
    \qedhere
  \end{align*}
\end{proof}

\subsection{Proofs of New Formulas}
\label{sub:new_formulas_proofs}

We now apply the theory developed in \cref{sub:hills_valleys}
to prove \cref{thm:waiting_monotonic, thm:residence_monotonic, thm:sesidence_monotonic}.
Throughout the proofs,
$\y$ and $\z$ refer to $\y[\generic{}]$ and~$\z[\generic{}]$, respectively.
Recall that $\ySize{}, \zSize{}$ form a hill-valley pair
(\cref{def:hill-valley_pair})
under any monotonic SOAP policy~\generic{}.

\begin{proof}[Proof of \cref{thm:waiting_monotonic}]
  \useamsalign\allowdisplaybreaks
  We compute
  \begin{align*}
    \mkern4mu&\mkern-4mu \E{\waiting{\generic{1}}{}}
    = \int_0^\infty \frac{\excess{\z}}{\coload{\y}\coload{\z}} \dF{x}
      \byref{eq:waiting_residence_monotonic}
    \\
    &\begin{aligned}[b]
    {}=
      \int_0^\infty \Biggl(
        & \frac{\F{\y}}{\coload{\y}}
          \Diff\sqgp[\bigg]{
            t \mapsto \frac{\excess{t}}{\coload{t}}
          }(\angle{\y, \z})
        \TwoColEqBreak[{}]
        + \frac{\F{\z} \excess{\z}}{\coload{\z}}
          \Diff\sqgp[\bigg]{
            t \mapsto \frac{1}{\coload{t}}
          }(\angle{\y, \z})
      \Biggr) \d{x}
      \iftwocol{
        \ifshort{}{\TwoColEqBreak[\hspace{14em}]}
        \byref{lem:yz_integrate_by_parts}
      }{}
    \end{aligned}
      \iftwocol{}{
        \ifshort{}{\TwoColEqBreak[\hspace{14em}]}
        \byref{lem:yz_integrate_by_parts}
      }
    \\
    &\begin{aligned}[b]
    {}=
      \int_0^\infty \Biggl(
        & \frac{\F{\y}}{\coload{\y}^2} \Diff\excess{\angle{\y, \z}}
        \ifshort{}{\TwoColEqBreak[{}]}
        + \frac{\F{\y} \excess{\z}}{\coload{\y}}
          \Diff\sqgp[\bigg]{
            t \mapsto \frac{1}{\coload{t}}
          }(\angle{\y, \z})
        \EqBreak[{}]
        + \frac{\F{\z} \excess{\z}}{\coload{\z}}
          \Diff\sqgp[\bigg]{
            t \mapsto \frac{1}{\coload{t}}
          }(\angle{\y, \z})
      \Biggr) \d{x}
      \iftwocol{
        \ifshort{}{\TwoColEqBreak[\hspace{14em}]}
        \byref{lem:Diff_split}
      }{}
    \end{aligned}
      \iftwocol{}{
        \ifshort{}{\TwoColEqBreak[\hspace{14em}]}
        \byref{lem:Diff_split}
      }
    \\
    &\begin{aligned}[b]
    {}=
      \int_0^\infty \Biggl(
        &
        \iftwocol{}{
          \frac{\F{\y}}{\coload{\y}\coload{\y}} \diff\excess{x}
          \TwoColEqBreak[{}] +
        }
        \excess{\z} \gp[\bigg]{
            \frac{\F{\y}}{\coload{\y}} + \frac{\F{\z}}{\coload{\z}}
          }
          \diff\sqgp[\bigg]{
            t \mapsto \frac{1}{\coload{t}}
          }(x)
        \iftwocol{
          \TwoColEqBreak[{}]
          + \frac{\F{\y}}{\coload{\y}\coload{\y}} \diff\excess{x}
        }{}
      \Biggr) \d{x},
      \iftwocol{\byref{lem:yz_Diff_diff}}{}
    \end{aligned}
      \iftwocol{}{\byref{lem:yz_Diff_diff}}
  \end{align*}
  which equals the desired result by~\cref{eq:coload_excess}.
\end{proof}

\begin{proof}[Proof of \cref{thm:residence_monotonic}]
  \useamsalign\allowdisplaybreaks
  We compute
  \begin{align*}
    \ifshort{&}{}\E{\residence{\generic{1}}{}}
    \ifshort{}{&}= \int_0^\infty \frac{x}{\coload{\y}} \dF{x}
      \byref{eq:waiting_residence_monotonic}
    \\
    &= \int_0^\infty \gp[\bigg]{
        \z \F{\z} \Diff\sqgp[\bigg]{
          t \mapsto \frac{1}{\coload{t}}
        }(\angle{\y, \z})
        + \frac{\F{x}}{\coload{\y}}
      } \d{x}
      \ifshort{}{\TwoColEqBreak[\hspace{18em}]}
      \byref{lem:yz_integrate_by_parts}
    \\
    &= \int_0^\infty \gp[\bigg]{
        \frac{-\z \F{\z}}{\coload{\y}\coload{\z}}
          \diff\coload{x}
        + \frac{\F{x}}{\coload{\y}}
      } \d{x},
      \byref{eq:Diff_recip, lem:yz_Diff_diff_chain}
  \end{align*}
  which equals the desired result by~\cref{eq:coload_excess}.
\end{proof}

\begin{proof}[Proof of \cref{thm:sesidence_monotonic}]
  \useamsalign\allowdisplaybreaks
  Very similarly to the proof of \cref{thm:residence_monotonic}, we compute
  \begin{align*}
    \ifshort{&}{}\E{\sesidence{\generic{1}}{}}
    \ifshort{}{&}= \int_0^\infty \frac{\z}{\coload{\y}} \dF{x}
      \byref{eq:sesidence}
    \\
    &= \int_0^\infty \gp[\bigg]{
        \z \F{\z} \Diff\sqgp[\bigg]{
          t \mapsto \frac{1}{\coload{t}}
        }(\angle{\y, \z})
        + \frac{\F{\y}}{\coload{\y}}
      } \d{x}
      \ifshort{}{\TwoColEqBreak[\hspace{18em}]}
      \byref{lem:yz_integrate_by_parts}
    \\
    &= \int_0^\infty \gp[\bigg]{
        \frac{-\z \F{\z}}{\coload{\y}\coload{\z}}
          \diff\coload{x}
        + \frac{\F{\y}}{\coload{\y}}
      } \d{x},
      \byref{eq:Diff_recip, lem:yz_Diff_diff_chain}
  \end{align*}
  which equals the desired result by~\cref{eq:coload_excess}.
\end{proof}

\end{document}

%% file: preamble.tex
\usepackage[nofancyhdr]{prelude}
\usepackage{threeparttable}

\newtoggle{short}
\newcommand{\ifshort}{\iftoggle{short}}

\input{figure_macros}

\newcommand{\captionsqueeze}{\ifshort{\vspace{-1.5em}}{}}
\newcommand{\slashfrac}[2]{#1/#2}

\newcommand{\proofin}[2][]{%
  \IfSubStr{#1}{,}{\begin{proof}[Proofs of \cref{#1}]}{\begin{proof}}
    See \cref{#2}. \noqed
  \end{proof}}

\newcommand{\EqBreak}[1][\quad]{\\* & #1}
\iftwocol{%
  \newcommand{\OneColEqBreak}[1][]{}
  \newcommand{\TwoColEqBreak}{\EqBreak}
}{%
  \newcommand{\OneColEqBreak}{\EqBreak}
  \newcommand{\TwoColEqBreak}[1][]{}
}

\newcommand{\EqMoveLeft}[1]{\iftwocol{\MoveEqLeft[1]}{\MoveEqLeft} #1 \ifshort{\\}{\\*}}
\iftwocol{%
  \newcommand{\OneColEqMoveLeft}[1]{#1}
  \newcommand{\TwoColEqMoveLeft}{\EqMoveLeft}
}{%
  \newcommand{\OneColEqMoveLeft}{\EqMoveLeft}
  \newcommand{\TwoColEqMoveLeft}[1]{#1}
}

\newcommand{\optional}[2]{%
  \ifempty{#2}{}{%
    \ifstrequal{#2}{[}{#1[}{%
      \ifstrequal{#2}{*}{#1*}{%
        #1{#2}%
      }}}}

\newcommand{\shortenRefs}{%
  \cCrefname{section}{Sec.}{Secs.}%
  \cCrefname{theorem}{Thm.}{Thms.}%
  \cCrefname{lemma}{Lem.}{Lems.}%
  \cCrefname{proposition}{Prop.}{Props.}%
  \cCrefname{corollary}{Cor.}{Cors.}%
  \cCrefname{definition}{Def.}{Defs.}%
  \renewcommand{\crefpairconjunction}{, }%
  \renewcommand{\crefmiddleconjunction}{, }%
  \renewcommand{\creflastconjunction}{, }%
  \renewcommand{\crefpairgroupconjunction}{, }%
  \renewcommand{\crefmiddlegroupconjunction}{, }%
  \renewcommand{\creflastgroupconjunction}{, }%
}

\newcommand{\byref}[1]{%
  \iftwocol{\quad}{&&}{\shortenRefs\text{\footnotesize[by \cref{#1}]}}}

\newcommand{\onetablenote}{%
  \renewcommand{\TPTnoteSettings}{%
    \labelsep0pt \leftmargin0pt \labelwidth0pt \footnotesize}}

\newcommand{\Diff}{\mathop{}\!\Delta}
\newcommand{\diff}{\mathop{}\!\partial}

\newsavebox{\invarrowbox}
\sbox{\invarrowbox}{%
  \begin{tikzpicture}[baseline=0cm, transform canvas={scale=0.7}]
    \draw[line width=0.4 / 0.7, <-] (0, 0) -- (0.115 / 0.7, 0);
  \end{tikzpicture}}

\newcommand{\F}{\ol{F}\optional\gp}
\newcommand{\dF}{\d{F}\optional\gp}
\def\G{\ol{F}_{\mkern-2mu e}\optional\gp}

\newcommand{\Ginv}{\ol{F}_{\mkern-2mu e}^{-1}\optional\gp}
\renewcommand{\H}[1][\rho]{H\optional\sb{#1}\optional\gp}
\newcommand{\coload}{\ol{\rho}\optional\gp}
\newcommand{\excess}{\tau\optional\gp}
\newcommand{\rank}[1]{r\optional\sp{#1}\optional\gp}

\newcommand{\response}[2]{T\optional\sp{#1}\optional\sb{#2}}
\newcommand{\waiting}[2]{Q\optional\sp{#1}\optional\sb{#2}}
\newcommand{\residence}[2]{R\optional\sp{#1}\optional\sb{#2}}
\newcommand{\sesidence}[2]{S\optional\sp{#1}\optional\sb{#2}}

\newcommand{\ySize}[2][]{y\optional\sp{#1}\optional\sb{#2}}
\newcommand{\zSize}[2][]{z\optional\sp{#1}\optional\sb{#2}}
\newcommand{\y}[1][]{\ySize[#1]{x}}
\newcommand{\z}[1][]{\zSize[#1]{x}}

\newcommand{\arelSize}[2][]{s\optional\sp{#1}\optional\sb{#2}}
\newcommand{\arel}[1][k]{\arelSize[\generic{#1}]{x}}

\newcommand{\textclass}{\textsf}
\newcommand{\NBUE}{\textnormal{\textclass{NBUE}}}
\newcommand{\IMRL}{\textnormal{\textclass{IMRL}}}
\newcommand{\DHR}{\textnormal{\textclass{DHR}}}

\newcommand{\ORname}{\textnormal{\textclass{OR}}}
\newcommand{\OR}[2]{\ORname(
  \ifstrequal{#2}{0}{0}{-#2},
  \ifstrequal{#1}{0}{0}{-#1}
)}
\newcommand{\Gumbel}{\textnormal{\textclass{MDA}}(\Lambda)}
\newcommand{\QIMRL}{\textnormal{\textclass{QIMRL}}}
\newcommand{\QDHR}{\textnormal{\textclass{QDHR}}}
\newcommand{\ENBUE}{\textnormal{\textclass{ENBUE}}}
\newcommand{\Bounded}{\textnormal{\textclass{Bounded}}}

\newcommand{\Qa}[1][]{%
  \textnormal{\textrm{I}}_Q\optional\sp{#1}}
\newcommand{\Qb}[1][]{%
  \textnormal{\textrm{II}}_Q\optional\sp{#1}}
\newcommand{\Rb}[1][]{%
  \textnormal{\textrm{II}}_R\optional\sp{#1}}
\newcommand{\Rc}[1][]{%
  \textnormal{\textrm{III}}_R\optional\sp{#1}}
\newcommand{\Sb}[1][]{%
  \textnormal{\textrm{II}}_S\optional\sp{#1}}
\newcommand{\Sc}[1][]{%
  \textnormal{\textrm{III}}_S\optional\sp{#1}}

\newcommand{\QRS}[1][]{%
  $\E{\waiting{#1}{}}$, $\E{\residence{#1}{}}$, and $\E{\sesidence{#1}{}}$}

\newcommand{\relwork}[2]{\mathrm{RelWork}^{#1}_{#2}}
\newcommand{\st}{\mathrm{st}}

\newcommand{\kmathy}[1]{%
  \ifmmode
    #1%
  \else
    \ifstrequal{#1}{k}{\lowercase{\textrm{\textit{k}}}}{#1}%
  \fi}

\makeatletter
\newcommand{\mg}[1]{M/G/\ifdefstring{\f@shape}{it}{}{\kern-0.08ex}\kmathy{#1}}
\makeatother

\newcommand{\policyNameFont}{\textrm}
\newcommand{\policyName}[3]{%
  \ifmmode%
    \textnormal{\policyNameFont{#2\ifempty{#3}{}{-}}}#3%
  \else%
    #1\ifempty{#3}{}{\=/\kmathy{#3}}%
  \fi%
}

\newcommand{\gittins}{\policyName{\gittinsName}{\gittinsMathName}}
\newcommand{\mgittins}{\policyName{\mgittinsName}{\mgittinsMathName}}
\newcommand{\serpt}{\policyName{\serptName}{\serptMathName}}
\newcommand{\mserpt}{\policyName{\mserptName}{\mserptMathName}}
\newcommand{\fb}{\policyName{FB}{FB}}

\newcommand{\srpt}{\policyName{SRPT}{SRPT}}
\newcommand{\rmlf}{\policyName{RMLF}{RMLF}}
\newcommand{\generic}{\policyName{$\pi$}{$\pi$}}
\newcommand{\mgeneric}{\policyName{M\=/$\pi$}{M\=/$\pi$}}

\newcommand{\gittinsName}{Gittins}
\newcommand{\mgittinsName}{M\=/Gittins}
\newcommand{\serptName}{SERPT}
\newcommand{\mserptName}{M\=/SERPT}

\newcommand{\gittinsMathName}{Gittins}
\newcommand{\mgittinsMathName}{M\=/Gittins}
\newcommand{\serptMathName}{SERPT}
\newcommand{\mserptMathName}{M\=/SERPT}

\newcommand{\lengthenMathNames}{%
  \def\gittinsMathName{\gittinsName}%
  \def\mgittinsMathName{\mgittinsName}%
  \def\serptMathName{\serptName}%
  \def\mserptMathName{\mserptName}%
}

\lengthenMathNames

%% file: fig_mgk.tex
\begin{tikzpicture}[semithick, scale=\ifshort{0.22}{0.31}]
  \begin{scope}
    \node at (7.76, 5.25) {\textsc{Single-Server System}};
    \draw (10.6, 0) circle (1.6);
    \node[below] at (10.6, -1.5) {speed $1$};
    \draw[->] (12.2, 0) -- ++(0.8, 0);
    \draw (4.75, -1.25) -- (9, -1.25) -- (9, 1.25) -- (4.75, 1.25);
    \draw (7.75, -1.25) -- ++(0, 2.5);
    \draw (6.5, -1.25) -- ++(0, 2.5);
    \draw (5.25, -1.25) -- ++(0, 2.5);
    \draw[->] (3.5, 0) node [left] {$\lambda$} -- ++(0.8, 0);
  \end{scope}
  \begin{scope}[shift={(\ifshort{18}{13}, 0)}]
    \node at (8.83, 5.25) {\textsc{$k$-Server System}};
    \newcommand{\server}[2]{%
      \draw (9, #1) -- (11, #2) -- (12.5, #2) ++(0.55, 0) circle (0.55)
        ++(0, -0.35) node [below] {\ifshort{\small}{}speed $1/k$} ++(0.55, 0.35) edge[->] ++(0.8, 0);}
    \server{0.98}{3.5}
    \server{0.21}{0.75}
    \filldraw (13.05, -1.4) circle (1.5mu) ++(0, -0.5) circle (1.5mu) ++(0, -0.5) circle (1.5mu);
    \server{-0.98}{-3.5}
    \draw (4.75, -1.25) -- (9, -1.25) -- (9, 1.25) -- (4.75, 1.25);
    \draw (7.75, -1.25) -- ++(0, 2.5);
    \draw (6.5, -1.25) -- ++(0, 2.5);
    \draw (5.25, -1.25) -- ++(0, 2.5);
    \draw[->] (3.5, 0) node [left] {$\lambda$} -- ++(0.8, 0);
  \end{scope}
\end{tikzpicture}


%% file: fig_serpt_example.tex
\renewcommand{\xscale}{180} 
\renewcommand{\yscale}{360} 

\begin{tikzpicture}[figure]
  \axes{1}{0.15}{$0$}{age~$a$}{$0$}{rank}

  \draw[primary] \mserptExampleRank;
  \draw[secondary] \serptExampleRank;

  \iftwocol{\newcommand{\xlegend}{1.05}}{\newcommand{\xlegend}{1.2}}
  \draw[primary] (\xlegend, 0.11) ++(-0.092, 0) -- ++(0.092, 0);
  \node[right] at (\xlegend, 0.11) {$\rank{\mserpt{}}{a}$};
  \draw[secondary] (\xlegend, 0.081) ++(-0.092, 0) -- ++(0.092, 0);
  \node[right] at (\xlegend, 0.081) {$\rank{\serpt{}}{a}$};
\end{tikzpicture}


%% file: fig_gittins_example.tex
\renewcommand{\xscale}{180} 
\renewcommand{\yscale}{360} 

\begin{tikzpicture}[figure]
  \axes{1}{0.15}{$0$}{age~$a$}{$0$}{rank}

  \draw[primary] \mgittinsExampleRank;
  \draw[secondary] \gittinsExampleRank;

  \iftwocol{\newcommand{\xlegend}{1.05}}{\newcommand{\xlegend}{1.2}}
  \draw[primary] (\xlegend, 0.11) ++(-0.092, 0) -- ++(0.092, 0);
  \node[right] at (\xlegend, 0.11) {$\rank{\mgittins{}}{a}$};
  \draw[secondary] (\xlegend, 0.081) ++(-0.092, 0) -- ++(0.092, 0);
  \node[right] at (\xlegend, 0.081) {$\rank{\gittins{}}{a}$};
\end{tikzpicture}


%% file: fig_overview.tex
\begin{threeparttable}
  \begin{tabularx}{37em}{@{}X@{}}
    \toprule
    \textbf{Key Definitions}
    \begin{itemize}[topsep=\medskipamount, itemsep=\medskipamount, after=\vspace{-\baselineskip}]
    \item
      (\cref{sub:distributions})
      \emph{Job size distribution classes:}
      $\QDHR$, $\OR{2}{\infty}$, $\Gumbel$, etc.
    \item
      (\cref{sec:overview, sec:mgk})
      \emph{Single-server quantities:}
      \QRS[\generic{1}].
    \item
      (\cref{sub:understanding_qrs})
      \emph{Age cutoffs:}
      $\y[\generic{}]$ and $\z[\generic{}]$.
    \end{itemize}
    \\ \midrule
    \textbf{Proof Steps}
    \begin{itemize}[topsep=\medskipamount, itemsep=\medskipamount, after={\vspace{-\medskipamount}\vspace{-\baselineskip}}]
    \item
      (\cref{sec:mgk})
      \emph{Compare \mg{k} to \mg{1}:}
      $\E{\response{\generic{k}}{}}
      \leq \E{\waiting{\generic{1}}{}}
        + k \E{\residence{\generic{1}}{}}
        + (k - 1) \E{\sesidence{\generic{1}}{}}$,
      whereas $\E{\response{\generic{1}}{}} = \E{\waiting{\generic{1}}{}} + \E{\residence{\generic{1}}{}}$.
    \item
      \emph{Show $\E{\waiting{\generic{1}}{}}$ dominates
        $\E{\residence{\generic{1}}{}}$ and $\E{\sesidence{\generic{1}}{}}$
        in $\rho \to 1$ limit.}
      \begin{itemize}[topsep=\medskipamount, itemsep=\medskipamount]
      \item
        (\cref{sec:rank_bounds})
        \emph{Job size distribution classes imply bounds on age cutoffs:}
        for example, if $X \in \QDHR$, then $\z[\generic{}] = O(x^\gamma)$ for some $\gamma \geq 1$.
      \item
        (\cref{sec:heavy})
        \emph{Job size distribution classes and bounds on age cutoffs
          imply $\E{\waiting{\generic{1}}{}}$ dominates:}
        for example, if $X \in \Gumbel$ and $\z[\generic{}] = O(x^\gamma)$ for some $\gamma \geq 1$,
        then $\E{\sesidence{\generic{1}}{}} = o({\waiting{\generic{1}}{}})$.
      \end{itemize}
    \item
      (\cref{sub:formal_proofs})
      \emph{Compare \mgittins{k} and \mserpt{k} to \gittins{1}.}
      \begin{itemize}[topsep=\medskipamount, itemsep=\medskipamount]
      \item
        \emph{\mgittins{k} vs. \gittins{1}:} prior work shows $\E{\waiting{\mgittins{1}}{}} \leq \E{\response{\gittins{1}}{}}$,
        implying $\lim_{\rho \to 1} \E{\response{\mgittins{k}}{}}/\E{\response{\gittins{1}}{}} = 1$.
      \item
        \emph{\mserpt{k} vs. \gittins{1}:} prior work shows $\E{\waiting{\mserpt{1}}{}} \leq 2 \E{\response{\gittins{1}}{}}$,
        implying $\lim_{\rho \to 1} \E{\response{\mserpt{k}}{}}/\E{\response{\gittins{1}}{}} \leq 2$.
      \end{itemize}
    \end{itemize}
    \\ \bottomrule
  \end{tabularx}
  \onetablenote
  \begin{tablenotes}
  \item
    Throughout, \generic{} stands for either \mgittins{} or \mserpt{}.
  \end{tablenotes}
\end{threeparttable}


%% file: fig_yz.tex
\begin{tikzpicture}[figure]
  \axes{15}{9}{$0$}{age}{$0$}{rank}

  \xguide[{$\y[\generic{}]$}]{5}{4}
  \yguide[$\rank{\generic{}}{x}$]{5}{4}
  \xguide[$x$]{6.5}{4}
  \xguide[{$\z[\generic{}]$}]{8.5}{4}

  \xguide[{$\ySize[\generic{}]{x'} = x' = \zSize[\generic{}]{x'}$}]{12.25}{7.75}
  \yguide[{$\rank{\generic{}}{x'}$}]{12.25}{7.75}

  \draw[primary]
  (0, 2) -- (3, 2)
  -- (5, 4) -- (8.5, 4)
  -- (13, 8.5) -- (15, 8.5);
\end{tikzpicture}


%% file: tab_yz.tex
\begin{tabular}{@{}llll@{}}
  \toprule

  \textsc{Size Dist\iftwocol{.}{ribution}} & \textsc{Quantity} & \textsc{Bound} & \textsc{Ref\ifshort{.}{erence}} \\

  \midrule

  $\OR{1}{\infty}$ &
  \ifshort{
  $\y[\mgittins{1}]$, 
  $\z[\mgittins{1}]$ &
  $\Theta(x)$ &
  \cref{thm:yz_bound_mgittins_or}\\
    }{
  $\y[\mgittins{1}]$ &
  $\Theta(x)$ &
  \cref{thm:yz_bound_mgittins_or}
  \\
  &
  $\z[\mgittins{1}]$ &
  $\Theta(x)$ &
  \\}
  &
  \ifshort{
  $\y[\mserpt{1}]$, 
  $\z[\mserpt{1}]$ &
  $\Theta(x)$ &
  \cref{thm:yz_bound_mserpt_or}\\
  }{
  $\y[\mserpt{1}]$ &
  $\Theta(x)$ &
  \cref{thm:yz_bound_mserpt_or}
  \\
  &
  $\z[\mserpt{1}]$ &
  $\Theta(x)$ &
  \\}

  \midrule

  $\QDHR$ &
  $\y[\mgittins{1}]$ &
  $\ifshort{\exists \gamma \geq 1,}{} \Omega(x^{1/\gamma})$\ifshort{}{ for some $\gamma \geq 1$} &
  \cref{thm:yz_bound_mgittins_qdhr}
  \\
  &
  $\z[\mgittins{1}]$ &
  $\ifshort{\exists \gamma \geq 1,}{} O(x^\gamma)$\ifshort{}{ for some $\gamma \geq 1$} &
  \\
  $\QDHR\iftwocol{}{\cup \QIMRL}$ &
  $\y[\mserpt{1}]$ &
  $\ifshort{\exists \gamma \geq 1,}{} \Omega(x^{1/\gamma})$\ifshort{}{ for some $\gamma \geq 1$} &
  \cref{thm:yz_bound_mserpt_qimrl}
  \\
  $\iftwocol{\ {} \cup \QIMRL}{}$ &
  $\z[\mserpt{1}]$ &
  $\ifshort{\exists \gamma \geq 1,}{} O(x^\gamma)$\ifshort{}{ for some $\gamma \geq 1$} &
  \\

  \bottomrule
\end{tabular}


%% file: tab_heavy.tex
\begin{tabular}{@{}llll@{}}
  \toprule

  \textsc{Size Dist\iftwocol{.}{ribution}} & \textsc{Quant\ifshort{.}{ity}} & \textsc{\ifshort{Bound}{Heavy-Traffic Scaling}} & \textsc{Reference} \\

  \midrule

  $\OR{1}{2}$ &
  $\E{\waiting{\generic{1}}{}}$ &
  $O\gp{-\log(1 - \rho)}$ &
  \cref{thm:heavy_mserpt_or_iv, thm:heavy_mgittins_response}
  \\
  &
  $\E{\residence{\generic{1}}{}}$ &
  $O\gp{-\log(1 - \rho)}$ &
  \\

  \midrule

  $\OR{2}{\infty}$ &
  $\E{\waiting{\generic{1}}{}}$ &
  $\ifshort{\exists \delta > 0,}{} \Omega\gp{(1 - \rho)^{-\delta}}$\ifshort{}{ for some $\delta > 0$} &
  \cref{thm:heavy_mserpt_or_fv, thm:heavy_mgittins_response}
  \\
  &
  $\E{\residence{\generic{1}}{}}$ &
  $O\gp{-\log(1 - \rho)}$
  \\
  &
  $\E{\sesidence{\generic{1}}{}}$ &
  $O\gp{-\log(1 - \rho)}$ &
  \cref{thm:heavy_mserpt_or_fv, thm:heavy_mgittins_sesidence}
  \\

  \midrule

 $\Gumbel$ &
  $\E{\waiting{\generic{1}}{}}$ &
  $\ifshort{\forall \epsilon > 0,}{} \Omega\gp{(1 - \rho)^{-(1 - \epsilon)}}$\ifshort{}{ for all $\epsilon > 0$} &
  \cref{thm:heavy_mserpt_gumbel, thm:heavy_mgittins_response}
  \\
  &
  $\E{\residence{\generic{1}}{}}$ &
  $\ifshort{\forall \epsilon > 0,}{} O\gp{(1 - \rho)^{-\epsilon}}$\ifshort{}{ for all $\epsilon > 0$} &
  \\
  $\Gumbel \cap \QDHR$ &
  $\E{\sesidence{\generic{1}}{}}$ &
  $\ifshort{\forall \epsilon > 0,}{} O\gp{(1 - \rho)^{-\epsilon}}$\ifshort{}{ for all $\epsilon > 0$} &
  \cref{thm:heavy_mserpt_gumbel, thm:heavy_mgittins_sesidence}
  \ifshort{}{\\
  $\Gumbel \cap \QIMRL$ &
  $\E{\sesidence{\mserpt{1}}{}}$ &
  $\ifshort{\forall \epsilon > 0,}{} O\gp{(1 - \rho)^{-\epsilon}}$\ifshort{}{ for all $\epsilon > 0$} &
  \cref{thm:heavy_mserpt_gumbel}}
  \\

  \midrule

  $\ENBUE$ &
  $\E{\waiting{\generic{1}}{}}$ &
  $\Theta\gp{(1 - \rho)^{-1}}$ &
  \cref{thm:heavy_mserpt_enbue, thm:heavy_mgittins_response}
  \\
  &
  $\E{\residence{\generic{1}}{}}$ &
  $\Theta(1)$ &
  \\
  $\Bounded$ &
  $\E{\sesidence{\generic{1}}{}}$ &
  $\Theta(1)$ &
  \cref{thm:heavy_mserpt_enbue, thm:heavy_mgittins_sesidence}
  \\

  \bottomrule
\end{tabular}


%% file: tab_notation.tex
\begin{tabular}{@{}lll@{}}
  \toprule

  \textsc{Notation} & \textsc{Description} & \textsc{Reference} \\

  \midrule

  \generic{k} &
  $k$-server version of SOAP policy \generic{} &
  \cref{sub:rank_functions} \\

  $\coload{a}, \excess{a}$ &
  functions of moments of $\min\{X, a\}$ &
  \cref{eq:coload_excess} \\

  $\y[\generic{}], \z[\generic{}]$ &
  new job and old job age cutoffs &
  \cref{def:yz} \\

  $\response{\generic{k}}{}$ &
  response time under \generic{k} &
  \cref{sub:rank_functions} \\

  $\waiting{\generic{1}}{}$ &
  waiting time under \generic{1} &
  \cref{eq:waiting_residence_monotonic} \\

  $\residence{\generic{1}}{}$ &
  residence time under \generic{1} &
  \cref{eq:waiting_residence_monotonic} \\

  $\sesidence{\generic{1}}{}$ &
  inflated residence time under \generic{1} &
  \cref{eq:sesidence} \\

  \bottomrule
\end{tabular}


%% file: fig_mserpt_serpt.tex
\begin{tikzpicture}[figure]
  \axes{15}{9}{$0$}{age}{$0$}{rank}

  \xguide[{$\y$}]{5}{4}
  \yguide[$\rank{\mserpt{}}{x}$]{5}{4}
  \xguide[$x$]{6.5}{4}
  \xguide[{$\z$}]{8.5}{4}

  \xguide[{$\ySize{x'} = x' = \zSize{x'}$}]{12.25}{7.75}
  \yguide[{$\rank{\mserpt{}}{x'}$}]{12.25}{7.75}

  \draw[primary]
  (0, 2) -- (3, 2)
  -- (5, 4) -- (8.5, 4)
  -- (13, 8.5) -- (15, 8.5);

  \draw[secondary]
  (0, 2) -- (1.5, 0.5) -- (3, 2)
  -- (5, 4) -- (6.75, 2.25) -- (8.5, 4)
  -- (13, 8.5) -- (15, 6.5);

  \begin{scope}[xscale=15, yscale=60]
    \iftwocol{\newcommand{\xlegend}{1.05}}{\newcommand{\xlegend}{1.2}}
    \draw[primary] (\xlegend, 0.11) ++(-0.092, 0) -- ++(0.092, 0);
    \node[right] at (\xlegend, 0.11) {$\rank{\mserpt{}}{a}$};
    \draw[secondary] (\xlegend, 0.081) ++(-0.092, 0) -- ++(0.092, 0);
    \node[right] at (\xlegend, 0.081) {$\rank{\serpt{}}{a}$};
  \end{scope}
\end{tikzpicture}
\\[\smallskipamount]
{\footnotesize Here $\y$ stands for $\y[\mserpt{}]$, and similarly for $\z$, $\ySize{x'}$, and $\zSize{x'}$.}


%% file: fig_yz_nonmonotonic.tex
\begin{tikzpicture}[figure]
  \axes{15}{9}{$0$}{age~$a$}{$0$}{rank}

  \xguide[{$\y[\generic{}] = \arel[1]$}]{3}{6}
  \xguide[$b$]{5}{4}
  \xguide[{$c = \arel[2]$}]{8}{4}
  \xguide[$x$]{10}{2}
  \xguide[{$\z[\generic{}]$}]{14}{6}
  \draw[guide] (5, 4) -- (8, 4);

  \draw[primary]
  (0, 3) -- (3, 6) -- (14, 6) -- (15, 7);

  \draw[secondary]
  (0, 3) -- (3, 6) -- (6.5, 2.5) -- (8, 4) -- (10, 2) -- (15, 7);

  \iftwocol{\newcommand{\xlegend}{16.9}}{\newcommand{\xlegend}{17.1}}
  \draw[primary] (\xlegend, 6.6) ++(-1.38, 0) -- ++(1.38, 0);
  \node[right] at (\xlegend, 6.6) {$\rank{\mgeneric{}}{a}$};
  \draw[secondary] (\xlegend, 4.86) ++(-1.38, 0) -- ++(1.38, 0);
  \node[right] at (\xlegend, 4.86) {$\rank{\generic{}}{a}$};
\end{tikzpicture}
